\documentclass[11pt,a4paper]{article}
\usepackage{jheppub,amsmath,  amssymb,slashed,url,bm,textgreek,upgreek,amsthm}
\usepackage{graphicx}
\usepackage{epstopdf}

\def\tt{{\mathfrak t}}

\def\sA{{\sf A}}
\def\sB{{\sf B}}
\def\sC{{\sf C}}

\def\u{{\sf u}}
\def\v{{\sf v}}

\def\b-theta{{\bm \theta}}

\def\v{{\eurm v}}

\def\u{{\eurm u}}

\def\SO{{\mathrm{SO}}}
\def\SU{{\mathrm{SU}}}

\def\Chi{\chi}

\def\be{\begin{equation}}
\def\ee{\end{equation}}

\def\Re{{\mathrm{Re}}}
\def\Im{{\mathrm{Im}}}

\def\tilde{\widetilde}

\def\sT{{\sf T}}

\def\O{{\mathcal O}}

\def\d{{\mathrm d}}

\def\b{\overline}
\def\rR{{\mathbb R}}
\def\C{{\mathbb C}}

\def\[{\bigl [}

\def\]{\bigr ]}
\def\S{{\sf S}}
\def\cN{{\mathcal N}}
\def\T{{\mathcal T}}

\def\zZ{{\mathbb Z}}

\def\W{{\mathcal W}}

\def\tilde{\widetilde}

\def\bar{\overline}

\def\Spin{{\mathrm{Spin}}}

\font\teneurm=eurm10 \font\seveneurm=eurm7  \font\fiveeurm=eurm5
\newfam\eurmfam
\textfont\eurmfam=\teneurm \scriptfont\eurmfam=\seveneurm
\scriptscriptfont\eurmfam=\fiveeurm
\def\eurm#1{{\fam\eurmfam\relax#1}}
\font\teneusm=eusm10 \font\seveneusm=eusm7 \font\fiveeusm=eusm5
\newfam\eusmfam
\textfont\eusmfam=\teneusm \scriptfont\eusmfam=\seveneusm
\scriptscriptfont\eusmfam=\fiveeusm

\font\tencmmib=cmmib10 \skewchar\tencmmib='177
\font\sevencmmib=cmmib7 \skewchar\sevencmmib='177
\font\fivecmmib=cmmib5 \skewchar\fivecmmib='177
\newfam\cmmibfam
\textfont\cmmibfam=\tencmmib \scriptfont\cmmibfam=\sevencmmib
\scriptscriptfont\cmmibfam=\fivecmmib

\def\AdS{{\mathrm{AdS}}}
\def\AAdS{{\mathrm{AAdS}}}
\def\la{\langle}
\def\ra{\rangle}
\def\Tr{\mathrm{Tr}}
\def\i{{\mathrm i}}

\newtheorem{theorem}{\rm\bf Theorem}[section]
\newtheorem{proposition}[theorem]{\rm\bf Proposition}

\newtheorem{lemma}[theorem]{\rm\bf Lemma}

\newtheorem{definition}[theorem]{\rm\bf Definition}

\title{No Ensemble Averaging \\ Below the Black Hole Threshold}

\author{Jean-Marc Schlenker$^1$  and Edward Witten$^2$}
\affiliation{$^1$ Department of Mathematics, University of Luxembourg  \\
$^2$ School of Natural Sciences, Institute for Advanced Study, \\ $~~$ 1 Einstein Drive, Princeton, NJ 08540 USA}

\abstract{In the AdS/CFT correspondence, 
amplitudes associated to connected bulk manifolds with disconnected boundaries have presented a longstanding mystery.   A possible interpretation is that they reflect
the effects of averaging over an ensemble of boundary theories.  But in examples in dimension $D\geq 3$, an appropriate ensemble of boundary theories does not
exist.    Here we sharpen the puzzle by
identifying a class of  ``fixed energy'' or ``sub-threshold'' observables that we claim do {\it not} show effects of ensemble averaging.   These are amplitudes that 
involve states that are above the ground state by only a fixed amount in the large $N$ limit, and in particular are far from being black hole states.
To support our claim, we explore the example of $D=3$, and show that connected solutions of Einstein's equations with disconnected boundary never contribute 
to these observables.   To demonstrate this requires some novel results about the renormalized volume of a hyperbolic three-manifold, which we prove
using modern methods in hyperbolic geometry.     Why then do any observables show apparent ensemble averaging?  We propose that this reflects the chaotic
nature of black hole physics and the fact that the Hilbert space describing a black hole does not have a large $N$ limit.}

\begin{document}\maketitle
\tableofcontents

\section{Introduction}\label{intro}

Since early days of the AdS/CFT correspondence \cite{Malda},   there has been a puzzle of how to interpret Euclidean amplitudes computed
on a connected bulk manifold $X$ whose conformal boundary $M$ is not connected \cite{WittenYau,MaldaLior}.    If $M$ is connected,
the sum over all choices of $X$ is interpreted in the AdS/CFT correspondence as computing what we will call $Z(M)$, the conformal field theory (CFT) partition
function on $M$.   What if $M$ has, say, two connected components $M_1,M_2$?   The sum over all choices of $X$ whose conformal boundary is
the disjoint union $M=M_1\sqcup M_2$ appears to compute, in some sense, a connected correlation function $\la Z(M_1)Z(M_2)\ra_{c}=\la Z(M_1)Z(M_2)\ra
-\la Z(M_1)\ra \la Z(M_2)\ra$ between the two CFT partition functions $Z(M_1)$ and $Z(M_2)$.     We will refer to such connected correlation functions
between different boundaries (possibly with operator insertions on the boundaries) as ``connected amplitudes with disconnected boundaries'' or CADB amplitudes for short.
CADB amplitudes are not  a standard concept in CFT or indeed in any Euclidean quantum field theory, so the fact that AdS/CFT duality
seems to give a way to calculate them has been puzzling.    

A new perspective came from the  discovery that a simple, soluble model, namely JT gravity in two dimensions,
 computes an ensemble average in a random matrix theory  \cite{SSS}.
A two-dimensional gravitational theory  is expected to have a dual description by an ordinary quantum theory on the 1-dimensional boundary of a 2-dimensional world.   A compact connected 1-manifold
is a circle $S$, say of circumference $\beta$, and the partition function is then $Z(\beta)=\Tr\,\exp(-\beta H)$, where $H$ is the Hamiltonian of the boundary
theory.   However, it turns out that the theory dual to JT gravity does not have a unique Hamiltonian $H$; rather, $H$ is drawn from a random matrix
ensemble.    This provides a rationale for the existence of connected correlation functions between observables associated to different boundary circles.   Such amplitudes can be generated in the boundary description by
averaging over the random matrix $H$. This discovery revived interest in much older ideas about gravitational wormholes and ensemble averages \cite{Coleman,GS,MM}.

But the interpretation of CADB amplitudes in terms of ensemble averages raises an immediate paradox.   In many basic examples
of AdS/CFT duality, it is believed that  the parameters on which the boundary theory depends are all  known and all have a known interpretation in terms of the bulk
theory.   The duality seems to say that a specific boundary theory, with specific values of the parameters, is dual to a specific bulk theory with the same parameters.
For instance, two of the original examples of AdS/CFT duality are the maximally supersymmetric models based on $\AdS_4\times \S^7$ and $\AdS_7\times \S^4$.
In these examples, the only parameter that the boundary CFT depends on is a single positive integer $N$.  The bulk theory also depends on $N$; in fact,
 Newton's constant $G$ varies as a negative power of $N$, and $N$ can be measured
as the integral of a certain seven-form or four-form on $\S^7$ or $\S^4$.   
The duality claims an equivalence between bulk and boundary theories for each choice of $N$; and for given $N$, there appears to be nothing else one could be averaging
over in an ensemble.   So how can CADB amplitudes in these theories be interpreted in terms of averaging over an ensemble?

We will try to shed some light on this question by arguing that when the Anti de Sitter dimension $D$ is at least 3, and therefore the boundary dimension
$d=D-1$ is at least 2, certain important  observables, which we will call fixed energy observables, are 
{\it not} affected by ensemble averaging; if one considers
only these observables, one will see no sign of ensemble averaging.    These are observables that can be defined just in terms of the energy and couplings
of states whose energy remains above the ground state by only a finite amount as $N$ becomes large.  
In particular these observables involve states that are far below the black hole threshold, which in the AdS/CFT correspondence is the energy of the Hawking-Page
transition \cite{HP} from a thermal gas in Anti de Sitter space to a black hole; in gauge theory examples,
this can also be interpreted as a deconfining transition \cite{Witten}.
The fixed energy observables are described more precisely in section \ref{prelims}.  
Another formulation of our proposal
involves integrability.   Some important examples of AdS/CFT duality are integrable in the large $N$ limit, and apparently also in an asymptotic
expansion around that limit; for an extensive review, see \cite{Beisert}.  Our proposal is that precisely the observables that are
accessible via integrability  are not affected by ensemble averaging.   This is consistent with the fact that calculations based on integrability do appear
to describe a definite CFT, not an ensemble average.   

In the case $D=3$, $d=2$, our proposal has implications for hyperbolic geometry that are explained in section \ref{implic}, so we can test the proposal
by verifying those predictions.      Thus, we consider examples of AdS/CFT
duality in which the bulk spacetime is asymptotic to $\AdS_3\times B$ for some compact manifold $B$.   The choice of $B$ will play no role in
our discussion; various examples  have been much studied, including
$\S^3\times {\sf T}^4$, $\S^3\times {\mathrm{K3}}$, and $\S^3\times \S^3\times \S^1$ ($\S^n$ is an $n$-sphere, $\sf T^n$ is an $n$-torus, 
and K3 is the complex  surface of that name).    Let $M_1, M_2,\cdots, M_n$ be a collection of Riemann surfaces.   If $X$ is a hyperbolic three-manifold whose conformal boundary
is a disjoint union $M_1\sqcup M_2\sqcup \cdots \sqcup M_n$, then a path integral on $X\times B$ contributes to a  connected amplitude
$\la Z(M_1)Z(M_2)\cdots Z(M_n)\ra_c$.     For small $G$ or equivalently large $N$, the contribution is proportional asymptotically to $\exp(-V_R(X)/4\pi G\ell^2)$,
where $V_R(X)$ is the renormalized volume of $X$, and $\ell$, which we assume much larger than $G$, is the AdS radius of curvature.   
We will deduce the statement that fixed energy amplitudes do not receive  contributions from manifolds with disconnected boundary
from properties of renormalized volumes.    Specifically, we will show that a hyperbolic three-manifold $X$  contributes to a fixed energy amplitude on a Riemann surface
$M$  only if the boundary of $X$ is connected and consists only of $M$.   An overview of the mathematical arguments is given in section \ref{overview}; details
appear in appendix \ref{details}.

Even if $X$ has connected boundary $M$, it is not necessarily true that $X$  contributes to fixed energy amplitudes on $M$.
To contribute to such amplitudes, $X$ must be a ``handlebody'' or Schottky manifold.  This means that for some embedding of $M$
in $\rR^3$, $X$ is topologically equivalent to the interior of $M$.   In other words, we will find that only the simplest hyperbolic three-manifolds with boundary $M$ contribute
to fixed energy amplitudes on $M$.

 We hope that these statements about hyperbolic three-manifolds are illustrating a general lesson that fixed energy
 observables are not subject to ensemble averaging, but  a number of caveats are necessary.    First,
we consider only the case $D=3$, $d=2$.    Though we  suspect that a similar picture holds for larger values of $D$, to show this will require more work 
with both the physics and the classical geometry.   Second, our main arguments concern the case that $X$ is a classical solution of Einstein's equations,
that is, a hyperbolic three-manifold.     However, it is believed that in some cases, non-solutions must be considered.\footnote{In at least some cases, these non-solutions can
be interpreted as critical points of the action at infinity  (see section 4.2 of \cite{WittenComplex} for discussion), as opposed to classical solutions, which are ordinary critical points. They can also possibly be interpreted, 
at least in some cases,
 as complex critical points (critical points of the analytic continuation of the Einstein action to a holomorphic function of a complex-valued metric tensor on $X$).}   Since little is known about what non-solutions are relevant in general, we cannot make a systematic analysis.
 However, we will consider the few examples that have been analyzed in the literature, namely $\rR\times {\sT}^2$, studied in \cite{CJ}, and Seifert fibered
 manifolds, studied in \cite{MaxTur}. 
The known and conjectured results are consistent with the claim that fixed energy observables do not show  apparent ensemble averaging.
 Finally, even if one restricts to classical solutions,
 there is no need to consider only spacetimes of the product form $X\times B$; one could consider a more general ten-manifold with the same asymptotic
 behavior.    Little is known about possible solutions of this more general form, and we are not in a position to prove that they do not contribute to fixed energy
 CADB amplitudes.   
 
Assuming  that fixed energy observables indeed
 do not have CADB contributions, this is a strong indication that in AdS/CFT duality, there is no ensemble average over unknown parameters.  It would be
 very hard for such an average not to affect fixed energy observables.
 Given this, why do amplitudes that involve black hole states have CADB contributions?   This question will be discussed in section
 \ref{why}.   We propose that this phenomenon reflects the following two points:   (1) black hole physics is highly chaotic \cite{Chaos}; (2) the Hamiltonian and Hilbert
 space that describe a single black hole apparently do not have a large $N$ or small $G$ limit.  
   The first of these assertions is well-known in the present
 context and was a large part of the motivation for the work that eventually led to the interpretation of JT gravity in terms of an ensemble average \cite{SSS}.
 The second assertion, which has not been considered in the present context, is a negative one; since the entropy of a black hole grows as a power of $N$
 (as $N^2$ in the case of $\cN=4$ super Yang-Mills in four dimensions), it is hard to see in what sense the Hilbert space describing a black hole
 could have a large $N$ limit, and the literature certainly does not contain any proposal for such a limit.   See \cite{WittenLecture} for more discussion.
  (By contrast, the thermofield double, which is dual to an entangled pair of black holes \cite{MaldaDouble},
 does have a large $N$ limit.)    Point (1) means that the Hamiltonian $H_N$ that describes black holes of a given energy for a given value of $N$
 can be viewed as a very
 large pseudorandom matrix, which will look like a random matrix in any standard calculation. 
 Point (2) suggests that for $N\not=N'$, even if $|N'-N|\ll N$, $H_N$ and $H_{N'}$  can be viewed to a good approximation as independent
 draws from a random matrix ensemble.   If so, then quantities that depend on $H_N$, such as the partition function, will be smooth functions of $N$ only
 to the extent that they are self-averaging in random matrix theory.  (In random matrix theory, a quantity is called self-averaging if it has almost the same value for almost any draw
 from a given  random matrix ensemble.)  
 Quantities that are not self-averaging in random matrix theory will  depend erratically on $N$
 or $G$.  
 With present techniques, what  we know how to calculate from the gravitational path integral are smooth functions of $N$ or $G$, and
 given the facts just stated, above the Hawking-Page transition, we can only calculate quantities that are self-averaging in random matrix theory.\footnote{If a quantity
 is not self-averaging but does have a nonzero average in random matrix theory, it may be possible to compute this average. Examples are discussed in \cite{SSS1}.}   
 If it is possible to compute erratically varying quantities from a gravitational path integral, this will involve an unfamiliar type of
  path integral that is not dominated by a simple sum over
  critical points,
 not even  critical points at infinity.   In random matrix theory, CADB amplitudes make sense and are sometimes self-averaging; when that is
 the case, there can be a simple way to calculate them from the gravitational path integral.     A shorthand way to summarize the proposal made here is that the ensemble averaging that is seen in gravitational path integral calculations is simply
 an averaging over nearby values of $N$ to eliminate  erratic fluctuations.   This makes sense as a general proposal because in all known examples of AdS/CFT
 duality for $D\geq 3$, $G$ varies inversely with one or more integers $N$.  
 
 We should stress again that we 
 unfortunately do not know for sure how to extrapolate from the specific results we will prove about hyperbolic three-manifolds to a general lesson
 about AdS/CFT duality.   As noted previously, in some important examples of AdS/CFT duality, states whose energy above the ground state is fixed for
 $N\to\infty$ are described by an integrable system.\footnote{\label{maldanote} As pointed out to us by J. Maldacena, even when there is no integrable 
 system, the spectrum of fixed energy states is never truly chaotic, since for large $N$ the dimension of a multi-trace operator
 $\Tr\,\O_1\,\Tr\,\O_2\cdots \Tr\,\O_k$ is simply the sum of the dimensions of the individual factors. This assumes
  that the number $k$ of factors is kept fixed for $N\to \infty$.
 For $k\sim N$, one enters a different regime in which nonlinear interactions are important in the large $N$ limit, potentially leading to classical and quantum chaos
 (though in general probably not the maximal chaos of black hole physics).}
Given this, and since  integrability is the antithesis of chaos, and given also the close relation of apparent ensemble averaging to chaos, the safest conjecture in the
context of the present article is that energies and couplings of  fixed energy states are not affected by apparent ensemble averaging.   Thus a conservative
 title for this article would be ``No Ensemble Averaging at Fixed Energy Above the Ground State.''   However, in practice, in a theory that can be well
 approximated by pure gravity up to the black hole threshold, the detailed results we obtain
 in section 3 are valid for any state that is  below that threshold; therefore, especially in section 3, when discussing classical solutions of pure gravity,
 we use the language of sub-threshold states, rather than fixed energy states.   We should note, though, that in a theory that can be approximated by pure gravity
 up to the black hole threshold, the sub-threshold states are all Virasoro descendants of the identity, which explains why their couplings behave similarly to those of the fixed
 energy states.
 
 After v1 of this paper was submitted to the arXiv, there appeared a very interesting article \cite{CCHM}            
  on couplings in 3d gravity of states whose energy, in the large $N$
 limit, is above the ground state by a fixed fraction -- positive but less than 1 -- of the energy required to make a black hole. (This regime might be similar to the
 regime mentioned in footnote \ref{maldanote} with $k\sim N$.)  Such states could be solitons, for example.   Couplings
 of such states do show apparent ensemble averaging.   A possible interpretation, in the spirit of section \ref{why}, is that couplings of these states
 are described in the semiclassical limit by the nonlinear gravity or supergravity theory, which (if it has the assumed states) is not integrable and is likely
 to lead to classical and quantum chaos.   
 
\section{Volumes and Fixed Energy Amplitudes}\label{hagedorn}

\subsection{Preliminaries}\label{prelims}

As was remarked in the introduction,  known examples of AdS/CFT duality in $D\geq 3$  always depend on one or more  integers with an inverse relation to
 Newton's constant
$G$.  For example, four-dimensional maximally supersymmetric Yang-Mills theory with gauge group $\SU(N)$ has a dual description in 
$\AdS_5\times \S^5$ with $G\sim 1/N^2$.   In  examples with $D=3$, $d=2$, which will be our main focus in this article, one 
has\footnote{This formula was actually discovered before the general understanding of
AdS/CFT duality \cite{BH}.}
\be\label{tofo} G=\frac{3\ell}{2c}, \ee
where $c$ is the central charge of the boundary CFT, and $\ell$ is the radius of curvature of the $\AdS_3$ space. 
For $D=3$, $d=2$,  the ``large $N$ limit'' is a limit in which $c$ is large, and therefore $G\ll \ell$; that last condition enables gravity to be treated 
semiclassically, by summing over classical solutions, as we will assume in this article.   
    For instance, in one much-studied family of models, $c=6Q_1 Q_5$, where $Q_1$ and $Q_5$ are integer-valued one-brane and five-brane charges.
In that example, by ``large $N$ limit,'' we mean the limit in which $Q_1$ and $Q_5$ are taken to be large, with a fixed ratio, ensuring, in particular,
that $c$ is large.

In a $d$-dimensional CFT, a  local operator $\O$ is inserted in a correlation function at a point $p$ in a $d$-manifold $M$.   A local operator $\O$ has a dimension
$\Delta$, which determines how it behaves under a conformal transformation that rescales the tangent space at $p$, and it transforms in an irreducible
 representation $J$
of the group $\SO(d)$ of rotations around the point $p$, or (in a theory with fermions) its double cover $\Spin(d)$. $J$ is called the spin of the operator.  In our
main example of $d=2$, the group $\Spin(2)$ is abelian and $J$ can be viewed as an integer or half-integer, an element of $\frac{1}{2}\zZ$.
In a CFT that participates in AdS/CFT duality, the dimensions $\Delta_i$ and representations $J_i$
of local operators $\O_i$ have a large $N$ limit.   This is a basic prediction of the duality,
and  in gauge theory examples in $d=4$ it is a consequence of the planar diagram expansion \cite{Thooft}.   There is precisely one local operator of dimension 0, namely
the identity operator 1, which transforms in a trivial 1-dimensional representation of $\Spin(d)$.  The other local operators $\O_i$, $i=1,2,\cdots$ have positive dimensions and can be labeled in order of increasing dimension  $0<\Delta_1\leq \Delta_2\leq \cdots$.

By the operator-state correspondence of CFT, local operators correspond to Hilbert space states if a CFT is quantized on a spatial manifold $\S^{d-1}$ (with a
round metric).   For odd $d$, the identity operator corresponds to a state of energy 0, but for even $d$ the identity operator corresponds to a state with an energy
that is determined by the anomaly in a conformal mapping from $\rR^d$ with a point removed to $\rR\times \S^{d-1}$.   We will be mainly interested in $d=2$,
in which case the identity operator corresponds to a state of energy
\be\label{negval} E_0=-\frac{c}{12}. \ee
Importantly, this value is negative and, in the large $N$ limit, it is large.   That will lead to a prediction that certain renormalized volumes of three-manifolds should
go to $-\infty$ in appropriate limits.
Any other operator $\O_i$ corresponds to a state  (more precisely a collection of states transforming in an irreducible representation $J_i$ of $\Spin(d)$) of energy
\be\label{negoval} E_i=E_0+\Delta_i.\ee
Thus the dimension $\Delta_i$ of an operator is the same as the excitation energy of the corresponding state above the ground state.  $\Delta_i$ is by definition the
energy of the $i^{th}$ excited state, and in the context of AdS/CFT duality, it has a limit for $c\to\infty$.   So in our terminology, the $i^{th}$ excited state for each
$i$ is a fixed energy state.

 $\Delta_i$ and $J_i$ are the first basic examples of fixed energy observables that we propose are not subject to ensemble averaging.
The other such observables are essentially the trilinear correlation functions of the $\O_i$.   The $\O_i$ can be normalized to put their two-point
functions in a standard form (for spinless operators, the standard form is  $\la \O_i(x) \O_j(y)\ra = \delta_{ij}/|x-y|^{2\Delta_i}$) and then the
trilinear or three-point correlation functions\footnote{In the case of spinless conformal primary fields, these three-point functions depend only on the
chosen points $x,y,z$.  More generally, one has to pick local parameters at $x,y,z$.   The details are not important in the present article.}  \be\label{threp}\lambda_{ijk}=\la \O_i(x)\O_j(y)\O_k(z)\ra\ee are important observables of a CFT.   We propose
that also the $\lambda_{ijk}$ are not affected by ensemble averaging.  

In $d=2$, all observables of a CFT are completely determined, in principle, in terms of the $\Delta_i$, $J_i$, and $\lambda_{ijk}$.
This can  be proved by using the fact that any two-manifold without boundary can be built by gluing together three-holed spheres, a fact that we will
exploit in section \ref{implic}.   As a result, in $d=2$, we expect that all observables not subject to ensemble averaging are actually determined
by $\Delta_i$, $J_i$, and $\lambda_{ijk}$.    In $d=2$, a complete set of conditions on $\Delta_i$, $J_i$, and $\lambda_{ijk}$ so that they are CFT
data is known in principle (but often hard to use in practice).  Above $d=2$, none of these statements have equally simple analogs, and in particular we do not
know whether to expect that the $\Delta_i$, $J_i$, and $\lambda_{ijk} $ are a complete set of observables that are free of ensemble averaging.

If $\Delta_i,J_i,$ and $\lambda_{ijk}$ are a complete set of CFT observables in $d=2$, and are not subject
to ensemble averaging, then why in $d=2$ do any observables appear to be subject to ensemble averaging?   The answer to this question depends on the fact that
the observables that appear to be subject to ensemble averaging are the ones that receive contributions from black hole states.   
We will refer to a spacetime asymptotic to $\AdS_3$ at spatial infinity as an $\AAdS_3$ spacetime.   The Einstein equations have an  $\AAdS_3$  black hole
solution, namely the BTZ black hole \cite{BTZ}.   It was understood in the original paper on the BTZ black hole that if energy is defined by the usual
ADM recipe of general relativity, then $\AdS_3$ itself has negative energy.   Later it was understood \cite{Strominger} that this negative energy can be understood in terms
of the central charge  of the BTZ black hole.   Thus $\AdS_3$ itself corresponds to the ground state of the CFT, with energy
\be\label{togo} E_0=-\frac{c}{12}, \ee with $c=3\ell/2G$ \cite{BH}.   Small perturbations of $\AdS_3$ give the states that have energy
$-c/12+\Delta_i$, with fixed $\Delta_i$ and large $c$.  
We get to the black hole regime if we take $c$ large with $\Delta\gg c/12$, meaning that the total
energy $E_0+\Delta$ is large and positive.   In that regime, the density of states per unit energy is exponentially large; it is $e^{S(E)}$, where $S(E)$, which is the
Bekenstein-Hawking entropy of the black hole at energy $E$, is of order $c$ for large $c$ and fixed $E/c$.

For any given value of $c$ or $N$, there is no useful notion of whether a given state is a black hole or not.
The black hole region is defined only in terms of a limiting process:  if we go to large $c$ with fixed $\Delta$, we get an ordinary state, and
if we go to large $c$ with  $E/c$ fixed and positive, we get a black hole.   The distinction is only sharp in the limit of large $c$, but because the gravitational calculations
that we know how to perform involve an asymptotic expansion at large $c$, the distinction is quite sharp in computations that we can actually perform.

 Concretely, the Hawking-Page phase transition occurs as follows.    In AdS/CFT duality, the partition function on a boundary manifold $M$ is computed
 by summing over bulk manifolds with conformal boundary $M$.   In stating the following, we will assume that $d=2$ and that we can assume the bulk
 manifold is a hyperbolic three-manifold $X$ with conformal boundary $M$.   (As explained in the introduction, in general there are  other possibilities.)
 The contribution of a given $X$ is, for small $G$, asymptotic to $\exp(-V_R(X)/4\pi G \ell^2)$.   That simple exponential is multiplied by an asymptotic
 series of quantum corrections; this series depends only on powers of $G$, not an exponential of $1/G$.   The choice of $X$  depends on the complex structure of $M$, and therefore so does
  the volume $V_R(X)$.   For example,
 if we are trying to compute $\Tr\,e^{-\beta H}$, we take $M$ to be a torus with a complex structure that depends on $\beta$ (see section \ref{reviewvol} for more detail).   Then the renormalized volume 
 $V_R(X)$ depends on $\beta$ and we can write it more explicitly as $V_R(X,\beta)$.    Since we have to sum over the choices of $X$ to compute
 $\Tr\,e^{-\beta H}$, we get 
 \be\label{wondo}\Tr\,e^{-\beta H}\sim \sum_\alpha \exp(-V_R(X_\alpha,\beta)/4G)\cdot F_\alpha(\beta,G), \ee
 where the sum runs over the possible choices of three-manifold $X_\alpha$, and for each $\alpha$, $F_\alpha(\beta,G)$ is the corresponding series of 
 quantum corrections.     
For any given value of $c$ or $G$, the sum in eqn. (\ref{wondo}) just produces an analytic function of $\beta$.   However, the asymptotic behavior
for $c\to\infty$ or $G\to 0$ is  dominated by the term in the sum with the smallest possible $V_R(X_\alpha,\beta)$.   As $\beta$ is varied, there can be
a ``crossover'' with a  jump in the choice of $X$ that minimizes $V_R(X,\beta)$ and therefore a discontinuous change in the asymptotic behavior of the partition function for small $G$.   That jump is the Hawking-Page transition.
What has just been explained (or its analog in higher dimensions) was actually the original explanation  of the Hawking-Page transition \cite{HP},
long before AdS/CFT duality was understood.
  
\subsection{Review Of The Renormalized Volume}\label{reviewvol}

The relation of the renormalized volume of a hyperbolic three-manifold to the negative ground state energy of a two-dimensional CFT will be important in what
follows, so we will review it in detail.   First we recall the relation between the Einstein action and the renormalized volume.

With negative cosmological constant, the  Einstein action in a three-dimensional spacetime $X$ of Euclidean signature, with boundary $\partial X$, is 
\be\label{einac}I=-\frac{1}{16\pi G}\int\d^3x \sqrt{g} \left(R+\frac{2}{\ell^2}\right)-\frac{1}{8\pi G}\int_{\partial X} \d^2x \sqrt \gamma K, \ee where $R$ is the Ricci
scalar of $X$,  and $\gamma$ 
and $K$ are the induced metric and the trace of the second fundamental form of $\partial X$ ($R_\gamma$ will denote the scalar curvature of $\gamma$). The last term is the Gibbons-Hawking-York (GHY) 
boundary term.  Using Einstein's equations $R_{ij}-\frac{1}{2} g_{ij} R-\frac{1}{\ell^2} g_{ij}=0$, one can rewrite the bulk part of the action as a multiple of the volume
of $X$:
\be\label{einac2}I=\frac{1}{4\pi G\ell^2}\int_X\d^3x \sqrt{g}-\frac{1}{8\pi G}\int_{\partial X} \d^2x \sqrt h K, \ee
 To this action, one can add ``counterterms,'' which are simply local integrals over $\partial X$ of invariant  functions of the induced metric.
For our purposes there are two relevant terms:
\be\label{zolgo}  \int_{\partial X}\d^2 x \sqrt \gamma \left(\frac{\alpha}{\ell^2}+\lambda R_\gamma\right), ~~~\alpha,\lambda\in\rR.\ee
One adjusts $\alpha$ and $\lambda$ to cancel the divergent part of the action (\ref{einac}) and to make the action conformally invariant, that is, invariant
under Weyl transformations of the boundary metric $\gamma$, apart from the usual $c$-number Weyl anomaly of two-dimensional CFT.
The bulk term in eqn. (\ref{einac2}) is formally $V/4\pi G$, where $V$ is the volume of $X$.  For a hyperbolic three-manifold with non-empty conformal boundary, this volume is divergent. After renormalization
it will be replaced with a renormalized volume $V_R$.   On the other hand, it turns out that in $D=3$, 
the GHY boundary term is entirely canceled by  renormalization.   So the renormalized action will be just  $I_R= V_R/4\pi G$.

A convenient reference on the necessary computation is \cite{HS}, which we will follow here (with minor changes of notation).   
To evaluate the action on an $\AAdS_3$ spacetime, one puts the metric of $X$ in the form
\be\label{acform} \d s^2= \frac{\ell^2}{4\rho^2}\d \rho^2 +\frac{1}{\rho} \sum_{i=1,2} \bar g_{ij}(x,\rho)\d x^i \d x^j \ee
near the conformal boundary of $X$, which in these coordinates is at $\rho=0$. 
In two dimensions, $\bar g(x,\rho) $ has an expansion
\be\label{wedform} \bar g(x,\rho)=\bar g_{(0)}(x)+\rho \bar g_{(2)}(x) +\rho\log\rho  \, h_{(2)}(x)+\O(\rho^2\log \rho). \ee
The only facts we need to know from the Einstein equations (eqn. (7) of \cite{HS}) are that 
\begin{align}\label{eform} \bar g_{(0)}^{kl} \bar g_{(2)kl}& = \frac{\ell^2}{2} R_{(0)} \cr
                      \bar g_{(0)}^{kl} h_{(2)kl}&= 0 ,\end{align}
  where $R_{(0)}$ is the Ricci scalar of the metric $g_{(0)}$.                     

The first step in defining the renormalized volume is to ``cut off'' $X$ by restricting to the 
region $X_\epsilon$ with $\rho\geq \epsilon$; to define $V_R$, one computes the volume of $X_\epsilon$ and then
 takes the limit $\epsilon\to 0$ after adjusting the counterterms of eqn. (\ref{zolgo}) to cancel
divergences.  Using the form (\ref{acform}) of the metric with the expansion (\ref{wedform}) together with (\ref{eform}),
one finds the divergent parts in the volume $V(X_\epsilon)$:
\begin{align}\label{divpart}V(X_\epsilon)&= \int_\epsilon\d \rho\int_{\partial X_\epsilon}
\d^2x\sqrt{\det g_{(0)}}\left(\frac{\ell}{\rho^2}+\frac{\ell^3}{4\rho}R_{(0)} +\O(\log\rho)\right)\cr
&= \frac{\ell}{\epsilon}\int_{\partial X_\epsilon}\d^2x \sqrt{\det g_{(0)}} +\frac{\ell^3 \log (1/\epsilon)}{4}\int_{\partial X_\epsilon} \d^2 x\sqrt{\det g_{(0)}}R_{(0)}
+{\mathrm{finite}}. \end{align}  Subtracting the divergent part, we arrive at the definition of the renormalized volume of $X$:\footnote{The conformal anomaly arises from the logarithmically divergent term.   If we repeat the calculation after making a Weyl rescaling
of the boundary metric $g_{(0)}$, the renormalized volume $V_R(X)$ is shifted by  the conformal anomaly.}
\be\label{renvol} V_R(X)=\lim_{\epsilon\to 0}\left( V(X_\epsilon)- \frac{\ell}{\epsilon}\int_{\partial X_\epsilon}\d^2x \sqrt{\det g_{(0)}} -\frac{\ell^3 \log (1/\epsilon)}{4}\int_{\partial X_\epsilon} \d^2 x\sqrt{\det g_{(0)}}R_{(0)}\right).\ee   

The GHY boundary term can be analyzed similarly.  One finds with the help of eqn. (\ref{eform}) that the GHY boundary term in the action of $X_\epsilon$
is a linear combination of
the two counterterms in eqn. (\ref{zolgo}) (plus a remainder that vanishes for $\epsilon\to 0$).   Hence the GHY boundary term does not contribute
to the renormalized Einstein action, which is just $I_R=V_R/4\pi G\ell^2$.

Now we can understand in terms of the renormalized volume the fact that $\Tr\,e^{-\beta H}$ diverges for large $\beta$ and large $c$
as $e^{\beta c/12}=e^{\beta\ell/8 G}$.   To compute $\Tr\,e^{-\beta H}$, we take the conformal boundary to be a two-torus $M$ parametrized
by $\phi$ and $t$ with metric $\d s^2=\d \phi^2+\d t^2$ and periodicities\footnote{$\phi$ should have period $2\pi$ because the statement that
the ground state energy of a CFT is $-c/12$ assumes that the CFT is quantized on a circle of circumference $2\pi$, here parametrized by $\phi$.
To compute $\Tr\,e^{-\beta H}$, we propagate the $\phi$ circle though imaginary time $\beta$, so we need $t\cong t+\beta$.}  $\phi\cong \phi+2\pi$ and $t\cong t+\beta$. 
To compute $\Tr\,e^{-\beta H}$, we have to sum over hyperbolic three-manifolds with boundary $M$.   The dominant one for large $\beta$ is
just $\AdS_3$ itself with a periodic identification $t\cong t+\beta$, making what we might call thermal $\AdS_3$.   The metric is\footnote{This metric is
often written in terms of $\tilde t=\ell t$.  Our normalization ensures that the metric on the torus at infinity is conformal to $\d t^2+\d \phi^2$.}
\be\label{adsm} \d s^2=\left(\frac{r^2}{\ell^2}+1\right)\ell^2 \d t^2 + \frac{\d r^2}{\frac{r^2}{\ell^2}+1} + r^2\d \phi^2.\ee
 To put this in the desired form of eqn. (\ref{acform}), we have to solve
 \be\label{changevar} \frac{\ell}{2}\frac{\d\rho}{\rho} =-\frac{\d r}{(\frac{r^2}{\ell^2}+1)^{1/2}},\ee
 leading to $\rho=\frac{1}{r^2}-\frac{\ell^2}{2r^4}+\cdots$ or
 \be\label{halpvar} r^2=\frac{1}{\rho} -\ell^2/2+{\mathcal O}(\rho).\ee
 The cutoff at $\rho\geq \epsilon$ therefore corresponds to $r\leq r_m= 1/\epsilon-\ell^2/2$, so the volume of $X_\epsilon$
 is
 \be\label{volx} V(X_\epsilon)=\int_0^{r_m}\d r\int_0^{2\pi} \d\phi \int_0^\beta \d t \sqrt{\det g}=\ell\int_0^{r_m}\d r\int_0^{2\pi} \d\phi \int_0^\beta \d t\, r
 =\pi\beta\left(\frac{\ell}{\epsilon}-\frac{\ell^3}{2}\right). \ee
 Subtracting the divergent part, we are left with $V_R(X)=-\pi\beta\ell^3/2$, so $\exp(-V_R/4\pi G\ell^2)=\exp(\beta\ell/8 G)$, as expected.   
 
 The moral of the story is that the large negative CFT ground state energy $-c/12$ for large $c$ corresponds to the fact that for large $\beta$, $V_R=-\pi \beta \ell^3/2$ becomes
very negative.   Generalizing this, our hypothesis that there is no ensemble averaging for fixed energy observables leads to predictions about precisely
 when $V_R(X)$, for a hyperbolic three-manifold $X$, can go to $-\infty$.   This will be explained in section \ref{implic}.   But first, we will discuss other contributions
 to $\Tr\,e^{-\beta H}$ to illustrate the fact that in most cases, $V_R(X)$ does not go to $-\infty$ when $\beta$ becomes large.
 
 Before being general, we will describe the special case that is actually important in understanding the Hawking-Page phase transition.   In a CFT, we are free to
 rescale the metric of the torus $M$, so instead of saying that $\phi$ and $t$ have periods $2\pi$ and $\beta$, we could rescale the metric and
 say that the periods are $4\pi^2/\beta$ and $2\pi$.   Now we can write down the same metric as before but with $t$ and $\phi$ exchanged:
 \be\label{nadsm} \d s^2=\left(\frac{r^2}{\ell^2}+1\right)\ell^2 \d \phi^2 + \frac{\d r^2}{\frac{r^2}{\ell^2}+1} + r^2\d t^2.\ee
 Assuming that $t$ is regarded as the Euclidean time direction, this is the Euclidean version of the BTZ black hole; we get the actual BTZ black
 hole
  if we continue to Lorentz signature\footnote{If we set $r^2+\ell^2=R^2$ and substitute $t\to \i t$, the line element (\ref{nadsm}) becomes $R^2\d\phi^2 + \ell^2\frac{\d R^2}{R^2-\ell^2}-(R^2-\ell^2)\d t^2$.
 This is a commonly written form of  the BTZ black hole metric (with zero angular momentum) \
 up to constant rescalings of $R,\phi$, and $t$.     Note that in this way of writing the BTZ metric,
  the black hole mass is encoded entirely in the period of the $\phi$ variable.  The black hole horizon is at $R=\ell$, where the coefficient of $\d t^2$ vanishes. As usual the vanishing of this coefficient represents only a breakdown of the coordinate system, and the Lorentz signature geometry has a real analytic continuation beyond this horizon.} by $t\to \i t$.   Obviously,  the renormalized volume can be evaluated just as before, but with 
 $\beta$ replaced by $4\pi^2/\beta$.  So the leading black hole contribution to $\Tr\,e^{-\beta H}$  is $\exp(\pi^2\ell/2\beta G)$, obtained from the previous
 $\exp(\beta\ell/8G)$ by $\beta\to 4\pi^2/\beta$.   Clearly, for $G$ asymptotically small, $\exp(\beta\ell/8G)$ and $\exp(\pi^2\ell/2\beta G)$ exchange dominance
 at $\beta=2\pi$, and this is the Hawking-Page phase transition.
 
 We can generalize this slightly to describe all hyperbolic three-manifolds whose conformal boundary is the torus $M$.  First, going back to $\AdS_3$,
 notice that we can slightly generalize the equivalence relation on $\phi$ and $t$ that we used before so that the torus at infinity is defined by
 \begin{align}\label{pairs} (\phi,t) & \cong (\phi+2\pi,t)\cr
     (\phi,t)& \cong (\phi+\alpha,t+\beta). \end{align}
A shift in imaginary time by $\beta$ is now accompanied by a rotation of the circle parametrized by $\phi$  by an angle 
$\alpha$, so a CFT path integral on this torus    computes
$\Tr\,e^{-\beta H+\i\alpha P}$, where $P$ is the operator that generates a rotation of the circle.   The eigenvalues of $H$ and $P$ on a state
that corresponds to an operator of dimension $\Delta$ and spin $J$ are $-\frac{c}{12}+\Delta$ and $J$, so
\be\label{newone} \Tr\,e^{-\beta H+\i\alpha P}=e^{\beta c/12} \sum_{i=0}^\infty e^{-\beta\Delta_i+\i \alpha J_i}. \ee
The computation of $V_R$ is not affected by the angle $\alpha$, so $V_R=-\pi\beta\ell^3/2$, independent of $\alpha$.   This agrees
with the fact that the exponentially growing factor $e^{\beta c/12}$ in the CFT partition function  (\ref{newone}) does not depend on $\alpha$, since the spin of the
ground state  is $J_0=0$.

 If we set $z=(\phi+\i t)/2\pi$,
     we see that the two-torus that is defined in eqn. (\ref{pairs}) is the quotient of the complex $z$-plane by the lattice generated by 1 and $(\i\beta+\alpha)/2\pi$.
     So it can be viewed as a complex elliptic curve with the modular parameter
     \be\label{modpar} \tau = \frac{\i\beta+\alpha}{2\pi}.\ee
    In this language, the formula for the renormalized volume becomes
    \be\label{renvolf} V_R(\tau)= -\pi^2\ell^3\Im(\tau). \ee
        
    Let $\u$ be the image in the boundary torus  $M$ of a straight line from $z=0$ to $z=1$, and let $\v$ be the image of a straight line from $z=0$ to $z=\tau$.
    Thus $\u$ and $\v$ are circles in $M$.   In the 
thermal $\AdS_3$ space of eqn. (\ref{adsm}), $\u$ is the boundary of a disc and $\v$ is not.\footnote{The coordinates $r,\phi$ can be viewed as polar coordinates
for a copy of $\rR^2$.  Then $\u$ is a large circle in this $\rR^2$ and therefore is the boundary of a disc.}   In the Euclidean black hole, since it is defined
by exchanging $t$ and $\phi$, the roles are reversed:
$\v$ is the boundary of a disc and $\u$ is not.   More generally, if $c,d$ are any relatively prime integers, there exists a hyperbolic three-manifold $X_{c,d}$ 
with boundary $M$ such that $c \u + d\v$ is the boundary of a disc in $X_{c,d}$ and other linear combinations are not.  
In this notation, $X_{0,1}$ is thermal $\AdS_3$ and $X_{1,0}$ is the Euclidean black hole.
 Generalizing the way we introduced
the Euclidean black hole, $X_{c,d}$ is defined by replacing the angles $2\pi t/\beta$ and $\phi$ with integer linear combinations of themselves.  In other words,
$X_{c,d}$ is actually the same manifold as the original thermal $\AdS_3$ space, but with the boundary parametrized differently.
The $X_{c,d}$ are actually the complete set of hyperbolic three-manifolds with boundary $M$. 
The renormalized volume of $X_{c,d}$ is
\be\label{renovol}V_R(X_{c,d}) =-\pi^2\ell^3 \frac{\Im(\tau)}{|c\tau+d|^2}. \ee
For our purposes, what is notable about this formula is that except in the original case $c=0,d=1$ of thermal $\AdS_3$, it never happens that $V_R\to-\infty$
for $\beta\to\infty$.   On the contrary, whenever $c\not=0$, eqn. (\ref{renvol}) implies that $V_R\to 0$ for $\beta\to\infty$.   Therefore, the thermal $\AdS_3$ space
$X_{0,1}$ is the only one of these manifolds that contributes to the partition function of the fixed energy states, or indeed to any states below the black hole
threshold.  

None of the observations in this section are in any way new.   The computation of the renormalized volume of thermal $\AdS_3$ and the BTZ
black hole is equivalent to the analysis of ADM masses in the original BTZ paper \cite{BTZ}.  As noted earlier, the CFT interpretation of the result
of this computation
goes back to \cite{Strominger}, and in turn to the construction of a boundary stress tensor \cite{BH}.
Summation over the manifolds $X_{c,d}$ has been 
considered in several previous papers -- in \cite{Farey,MaMo} to count certain
supersymmetric black hole states in $\AdS_3$, and in \cite{MaWi} in an attempt to construct a partition function
of three-dimensional pure gravity.  The special role of $X_{0,1}$ was part of those analyses.

\subsection{Implications of The Conjecture For Hyperbolic Geometry}\label{implic}

We have proposed that certain observables are not affected by ensemble averaging: the excitation energies $\Delta_i$ and spins $J_i$ of
the fixed energy states, and their trilinear couplings $\lambda_{ijk}$.  Those couplings can be computed as three-point functions 
$\la \O_I(x)\O_j(y)\O_k(z)\ra$ on $\S^2$.  Another description will be more useful in what follows: because
 of the operator-state correspondence of CFT, the $\lambda_{ijk}$  can also be computed by a path
integral on a three-holed sphere, with the fixed energy states $i,j,k$ inserted on its boundaries (fig. \ref{ThreeHoled}).  

 \begin{figure}
 \begin{center}
   \includegraphics[width=2.5in]{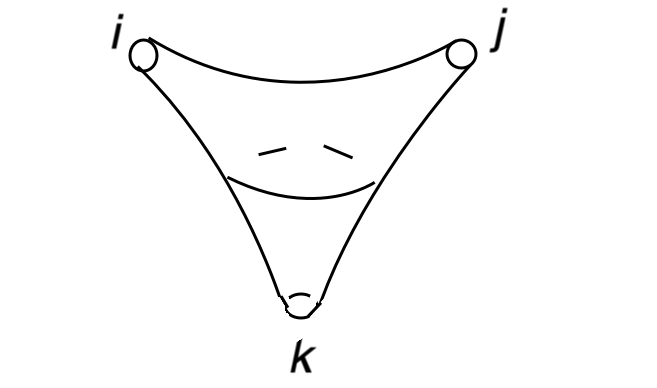}
 \end{center}
\caption{\small A three-holed sphere, with its boundaries labeled by fixed energy states $i,j,k$.
 \label{ThreeHoled}}
\end{figure} 

 \begin{figure}
 \begin{center}
   \includegraphics[width=3.5in]{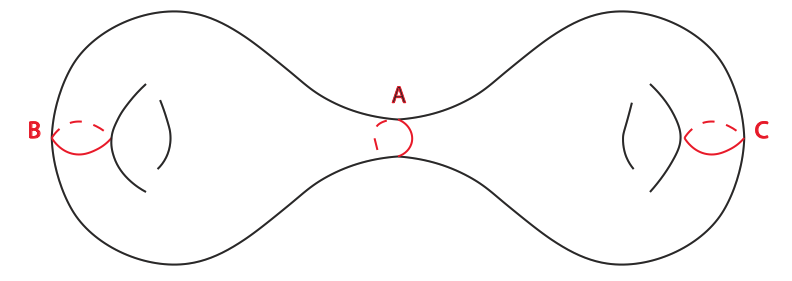}
 \end{center}
\caption{\small A genus two surface $M$ with three nonintersecting and homotopically independent one-cycles labeled $\sA$, $\sB$, and $\sC$.
$M$ is the conformal boundary of a hyperbolic three-manifold $X$.   If $X$ is such that $V_R(X)\to -\infty$ when $\sA$ is pinched, then the contribution of $X$ to $Z(M)$
has a part that describes propagation through $\sA$ of a fixed energy state.    In general, 
however, black hole states are propagating through $\sB$ and $\sC$ and this amplitude is subject to ensemble averaging.  If we stipulate
that $V_R(X)\to -\infty$ when any of $\sA$, $\sB$, or $\sC$ is pinched, we get an amplitude that has a contribution  that describes fixed energy states
propagating through each of $\sA$, $\sB$, and $\sC$, and interacting via two three-holed spheres, one to the left of $\sA$ and one to the right.
  A contribution of this type, according to our conjecture, should not be subject to ensemble averaging.   So this behavior of $V_R(X)$ should be
  possible only if the conformal boundary of $X$ is connected.
 \label{GenusTwo}}
\end{figure} 

In order to test our proposal using properties of hyperbolic three-manifolds, we want to identify observables in genus $g\geq 2$ that
can be determined in terms of $\Delta_i$, $J_i$, and $\lambda_{ijk}$.  Hopefully, the hyperbolic geometry will work out in such a way that
amplitudes  that
can be determined in terms of $\Delta_i$, $J_i$, and $\lambda_{ijk}$ receive contributions only from hyperbolic manifolds with connected boundary.

 To illustrate the idea, drawn in fig. \ref{GenusTwo} is a genus 2 surface $M$,
along with three
 nonintersecting and homotopically independent  one-cycles $\sA$, $\sB$, and $\sC$.   From a conformal point of view, it is equivalent to say
 that a one-cycle such as $\sA$ is becoming short, or is being ``pinched,'' or that the tube through $\sA$ is becoming long.   (We used this equivalence in 
 section \ref{reviewvol} when we said that it is equivalent conformally  to consider the circle parametrized by $t$ to be long, of circumference $\beta\gg 1$,
 or to consider the circle parametrized by $\phi$ to be short, of circumference $4\pi^2/\beta\ll 1$.)   When $\sA$ is pinched or the tube is long, the CFT
 partition function $Z(M)$ grows exponentially, because of the negative CFT ground state energy.

Now let $X$ be a hyperbolic three-manifold that has $M$ in its conformal boundary.   If the conformal boundary of $X$ consists only of $M$,
then the path integral $Z_X$  contributes to $Z(M)$.   If the conformal boundary is $M\sqcup M'$ for some $M'$ (which may or may not be connected),
 then $Z_X$  contributes to a connected correlation function $\la Z(M) Z(M')\ra_c$.  In either case, we ask what happens when
 a one-cycle such as $\sA$ or $\sB$ or $\sC$ is pinched.   Does $Z_X$ grow exponentially, reflecting in terms of the boundary CFT a
 sub-threshold state propagating through the cycle that is being pinched?\footnote{In order for a state of fixed energy above the ground state
  to propagate through the cycle that is being
 pinched, we need a stronger condition that $Z_X$ grows as $e^{\beta c/12}$.  It turns out, however, that the interesting constraints on $X$ arise if one merely
 asks for exponential growth of $Z_X$, without specifying the rate, and therefore it will not be important to distinguish fixed energy states from sub-threshold states.}  
 
 Since $Z_X$ is asymptotic for small $G$ to 
 $\exp(-V_R/4\pi G)$, a necessary condition for $Z_X$ to show exponential growth in the pinching limit is that $V_R(X)$ must go to $-\infty$ in that
 limit.   In section \ref{overview} and appendix \ref{details} we determine the condition on $X$ such that $V_R(X)$ goes to $-\infty$ when
 a given boundary cycle $\sA$ is pinched.   The answer is that this occurs if and only if $\sA$ is the boundary of a disc in $X$.   
 This generalizes the previously known facts for genus 1 that were summarized in section \ref{reviewvol}: of the manifolds $X_{c,d}$, the only
 one with the property that $\lim_{\beta\to\infty}V_R(X_{c,d})=-\infty$ is the thermal $\AdS_3$ space $X_{0,1}$, and this is also the only one
in which the circle parametrized by $\phi$ (which is the one that is pinched for $\beta\to\infty$) is the boundary of a disc.  Mathematically, if a circle $\sA$ in a component
$M$ of the 
conformal boundary of $X$ is the boundary of a disc in $X$ (but not in $M$), then $\sA$ is said to be ``compressible'' in $X$.   If  $M$ 
contains a compressible circle, then $M$ itself is said to be compressible.

In most cases, a given boundary circle $\sA$  is not compressible, so $Z_X$ does not contribute to an amplitude in which a sub-threshold state propagates through
$\sA$.  For many choices of $X$, there is no compressible circle at all in a given boundary component $M$.   (An example is the Fuchsian manifold that we discuss later.)
 But even if $\sA$ is compressible in $X$, $X$ typically does not contribute to an amplitude on $M$ that we expect to be free of
ensemble averaging, because even if a sub-threshold state is propagating through $\sA$, there may be black hole states propagating in the rest of the Riemann
surface.   For example, in fig. \ref{GenusTwo}, even if a sub-threshold state is propagating through $\sA$, 
the states propagating on the genus one Riemann
surfaces to the left and right of $\sA$  may be black hole states.  
To identify something that we expect to be free of ensemble averaging, we reason as follows.   
If we cut the genus two surface $M$ on the three nonintersecting and homotopically independent  one-cycles $\sA$, $\sB$, and $\sC$, it decomposes into the
union of two three-holed spheres.   A contribution to $Z(M)$ in which specified sub-threshold states $i,j,k$ are propagating through $\sA$, $\sB$,
and $\sC$ can be evaluated in terms of the dimensions and spins of the three states and the path integrals on the two three-holed spheres,   The boundaries
of the two three-holed spheres are labeled by the particular sub-threshold states that propagate through $\sA$, $\sB$, and $\sC$ 
 (in the example sketched in the figure, if states $i,j,k$ propagates respectively through $\sA$, $\sB$, and $\sC$, then the labels are $ijj$ for the three-holed
sphere on the left and $ikk$ for the three-holed sphere on the right).   The path integral on such a labeled three-holed sphere computes a trilinear
coupling of sub-threshold states.
Since, according to our conjecture, the dimensions and spins of the
sub-threshold states and their trilinear couplings are all unaffected by ensemble
averaging, we expect that a contribution to $Z(M)$ with sub-threshold states propagating through $\sA$, $\sB$, and $\sC$ is not subject to ensemble averaging.
On the other hand, a 
necessary condition for  $X$ to contribute to an amplitude with sub-threshold states propagating through each of $\sA$, $\sB$, and $\sC$ is that
$V_R(X)$ must go to $-\infty$ when any of $\sA$, $\sB$, or $\sC$ is pinched.  
Equivalently, in view of what has already been said, $\sA$, $\sB$, and $\sC$ must all be compressible in $X$. 
Conversely, if $\sA$, $\sB$, and $\sC$ are all compressible in $X$, then $V_R(X)$ does go to $-\infty$ when any of those cycles is pinched.
In that case, we expect that the conformal boundary of $X$ is connected and consists only of $M$, ensuring that $X$ does not contribute to a connected amplitude between
disconnected boundaries.  

This discussion has a straightforward generalization to higher genus.  If a surface $M$ of genus $g\geq 2$
 is cut along $3g-3$ nonintersecting and homotopically independent  circles $\sA_\sigma$, $\sigma=1,\cdots,
3g-3$, then it decomposes into
a union of $2g-2$ three-holed spheres.   A contribution to $Z(M)$ in which a sub-threshold state is propagating through each of the $\sA_\sigma$, should,
according to our conjecture, not be subject to ensemble averaging.   For a hyperbolic manifold $X$ whose  conformal boundary contains $M$ 
to contribute to such an amplitude, the $\sA_\sigma$ must all be compressible in $X$.   So we expect that if $M$ contains $3g-3$ nonintersecting and homotopically
independent
circles $\sA_\sigma$ that are all compressible in $X$, then the conformal boundary of $X$ consists only of $M$.

As we explain in section \ref{overview} and appendix \ref{details}, this is true and in fact more is true.    If at least $3g-5$ homotopically independent and
non-intersecting one-cycles in $M$ are compressible in $X$ (or $3g-6$ of them for $g\geq 3$), 
then the conformal boundary of $X$ is connected and consists only of $M$, and moreover
$X$ is a Schottky manifold.   A Schottky manifold is the simplest type of hyperbolic three-manifold.   A Schottky manifold with conformal boundary $M$ is,
topologically, the ``interior'' of $M$ for some embedding of $M$ in $\rR^3$.   For example, the picture drawn in fig. \ref{GenusTwo} suggests an embedding of
a genus two surface in $\rR^3$.  The interior of $M$ for this embedding is a three-manifold in which $\sA$, $\sB$, and $\sC$ are all compressible.  Topologically,
this interior is  called a handlebody.   A handlebody with boundary $M$ admits a hyperbolic metric with $M$ as the conformal boundary (for any choice
of the conformal structure on $M$).   Endowed with such a metric the handlebody is called a Schottky manifold.

From what has just been explained, we learn, modulo the arguments in section \ref{overview} and appendix \ref{details},
 that three-dimensional hyperbolic geometry, at least for the questions that we have asked, is consistent with our hypothesis that
certain AdS/CFT observables are not affected by ensemble averaging.  

Some of the points can be illustrated with the simple example of a Fuchsian manifold.  This is a three-manifold that is
topologically of the form $X=M\times \rR$, where $M$ is a Riemann surface of genus $g\geq 2$.   $X$ carries a hyperbolic metric of the form\footnote{\label{convcore}
The
submanifold $X_C$ defined by $t=0$ is totally geodesic and moreover is geodesically convex (any geodesic in $X$ between two points in $X_C$ is
actually contained in $X_C$).   $X_C$ is called the ``convex core'' of $X$, a notion that will be important in section \ref{overview}.  This example is exceptional,
because $X_C$ has volume 0.  Apart from a Fuchsian manifold, or a solid torus, such as the BTZ black hole, the convex core of any other hyperbolic three-manifold, including the quasi-Fuchsian ones that are discussed momentarily, 
has positive volume. (The convex core is not defined for $\AdS_3$ itself.)} 
\be\label{expmet}\d s^2 = \ell^2\left( \d t^2 +\cosh^2 t\, \d\Omega^2\right),\ee
where $\d\Omega^2$ is a metric on $M$ of constant scalar curvature $R=-2$.   The conformal boundary of $X$ consists of two copies of $M$, at
$t=-\infty$ and $t=+\infty$, respectively; let us call these $M_1$ and $M_2$.  $M_1$ and $M_2$ have the same complex structure since they have the same
metric $\d\Omega^2$.     The path integral on $X$  contributes to a connected amplitude
$\la Z(M_1)Z(M_2)\ra_c$.   However, a simple calculation shows that $V_R(X)$ is a topological invariant, independent of the complex structure of $M_1$ and
$M_2$.   $V_R$ never goes to $-\infty$ regardless of how we vary the complex structure of the boundary, so the path integral on $X$ never shows
the exponential growth characteristic of a sub-threshold state.

The drawback of this simple example is that since $M_1$ and $M_2$ have the same metric, we cannot vary their complex structures independently.
What happens, for example, if we pinch a cycle in $M_1$ without changing $M_2$?   In fact, there is a more general family of
hyperbolic metrics on $X$ such that the complex structures on $M_1$ and $M_2$ do vary independently.  These metrics are called quasi-Fuchsian,
and unfortunately in a situation that involves pinching on only one side,
they are only known by existence proofs, not explicit formulas.   At any rate, no one-cycle in either $M_1$ or $M_2$ is compressible,
so the results of section \ref{overview} and appendix \ref{details} imply that $V_R(X)$ never goes to $-\infty$.   Thus, although $X$ does contribute
to a connected amplitude $\la Z(M_1)Z(M_2)\ra_c$, its contribution only involves black hole states,  with no sub-threshold states  appearing in any channel.

So far in this section, we have only considered the case of a surface of genus $\geq 2$, and in fact the case of genus 1 is exceptional.   
A hyperbolic three-manifold whose conformal boundary contains a component of genus 1 is always one of the manifolds $X_{c,d}$ that were discussed in 
section \ref{reviewvol}.   In particular, at the level of hyperbolic geometry there are no disconnected amplitudes involving a torus ${\sf T}^2$.   
This is, however, not the whole story.   There is good reason to suspect  \cite{SSS1}  the existence in AdS/CFT models of a connected
correlation function $\la Z(M)Z(M')\ra_c$  where $M$ and $M'$ are tori, even though  there are no  classical hyperbolic manifolds
that can generate such a contribution. 
It has been argued  in \cite{CJ} that a path integral  on $\rR\times {\sf T}^2$ contributes to  $\la Z(M)Z(M')\ra_c$, even though no classical solution
is available on this manifold.  That paper actually contains in eqn. (3.56) an interesting formula for the contribution of $\rR\times {\sf T}^2$ to 
$\la Z(M)Z(M')\ra_c$, with independent complex structures on $M$ and $M'$.  We do not know if this formula is precisely correct.   However, the formula
is independent of $G$ and in particular shows no exponential growth for $G\to 0$.  So if it is even qualitatively correct, the contribution of $\rR\times {\sf T}^2$
to  $\la Z(M)Z(M')\ra_c$ is completely consistent with out conjectures.

There is one other case known of manifolds that do not admit classical hyperbolic metrics but nonetheless make relatively well understood contributions
to the path integral of three-dimensional gravity with $\Lambda<0$.  These are Seifert fibered manifolds, whose contributions were analyzed in \cite{MaxTur}
in the Kaluza-Klein limit, that is, the limit that the fiber is small.  As the remaining moduli are varied, these path integrals do not show contributions of sub-threshold states.

\section{Overview of Mathematical Arguments}\label{overview}

In the following, $M$ is an oriented two-manifold that is a component of the conformal boundary of a hyperbolic three-manifold $X$.   
It is known by elementary arguments that if $M$ has genus 0, then $X$ must be $\AdS_3$ itself, and if $M$ has genus 1, $X$ is one of the manifolds
$X_{c,d}$ that were introduced in section \ref{reviewvol}.   Therefore, we are primarily interested in the case that $M$ has genus at least 2.   Such
an $M$ admits a hyperbolic metric of constant scalar curvature $R=-2$; the space of such metrics, up to diffeomorphisms of $M$ that are isotopic
to the identity, is the Teichm\"{u}ller space $\T(M)$.   As one approaches the boundary of $\T(M)$, it is possible for the length of a simple closed geodesic $\gamma$
to go to 0; we will say that in that limit, $\gamma$ is ``pinched.'' 
We say that a collection of
embedded circles in $M$ are ``independent'' if they are non-intersecting and homotopically independent.   If $M$ has genus $g\geq 2$, a maximal set of independent circles
in $M$ consists of  $3g-3$ circles, as in fig. \ref{GenusThree}.   If $g=1$, a maximal set consists of just one circle.   An embedded
 circle in $M$ is compressible in $X$ if it is homotopically nontrivial in $M$ but bounds a disc in $X$.   As explained in more detail in section \ref{implic},
 a Schottky manifold is topologically the ``interior'' of $M$, for some embedding of $M$ in $\rR^3$.  Topologically the   $X_{c,d}$ are solid tori, the
 genus 1 analogs of a Schottky manifold.

The arguments of section \ref{hagedorn} relied on two mathematical assertions:

\begin{proposition}  \label{pr:3g-3}  If a component $M$ of the conformal boundary of $X$
contains a collection of at least $3g-5$ independent circles that are compressible in $X$, then the boundary of
$X$ is connected and consists only of $M$, and moreover $X$ is a Schottky manifold.\end{proposition}

\begin{theorem}\label{tm:main}  In the limit that a simple closed geodesic $\gamma\subset M$ is pinched, the renormalized volume $V_R(X)$ remains finite
if $\gamma$ is not compressible in $X$, and approaches $-\infty$ if $\gamma$ is compressible in $X$. \end{theorem}

Here we will give a very rough sketch of the proofs.\footnote{The bound $3g-5$ in Proposition \ref{pr:3g-3} is stronger than
we actually needed in section \ref{hagedorn}, where a bound of $3g-3$ would have sufficed.     We do not know if the stronger bound
is significant in the AdS/CFT context.  In appendix \ref{ssc:maximal}, we prove a slightly sharper bound for $g\geq 3$ (Theorem \ref{tm:3g-6}). Likewise
Theorem \ref{tm:main} is sharpened in Theorem \ref{tm:asymptotics}, as remarked at the end of this section.}    Further detail is explained in appendix \ref{details}.

 \begin{figure}
 \begin{center}
   \includegraphics[width=4in]{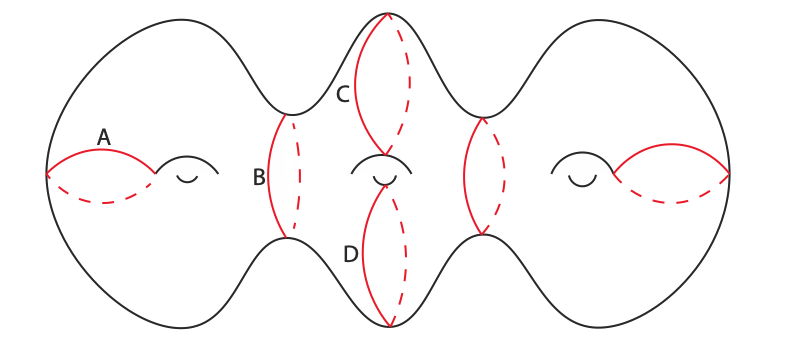}
 \end{center}
\caption{\small A  surface $M$ of genus $g=3$, with $3g-3=6$ independent circles marked.  Surgery on $\sf A$ produces a surface $M_1$ of genus 
$g_1=g-1=2$.   In passing to $M_1$, $\sf A$ disappears, $\sf B$ becomes nullhomotopic, and $\sf C$ and $\sf D$ become homotopic to each other,
leaving three independent circles on $M_1$.   Surgery on $\sf B$ leaves a surface $M_1'$ of genus $g_1'=1$ and a surface $M_1''$ of genus $g_1''=2$.
Of the original 6 independent circles on $M$, $\sf B$ disappears in the surgery, $\sf A$ remains on $M_1'$, and $\sf C$ and $\sf D$ become homotopic.
We remain with $3g_1''-3=3$ independent circles on $M_2'$.  
 \label{GenusThree}}
\end{figure} 

One important tool in the arguments is that, by varying the hyperbolic metric of $X$, the conformal structure of $M$ can be varied arbitrarily.
Indeed, by a classic result, the moduli space of hyperbolic metrics on $X$
is the  space of all complex structures on the conformal boundary of $X$ (whether that conformal boundary is connected or not)
modulo diffeomorphisms of the conformal boundary that extend over $X$.  

Therefore, in discussing Proposition 3.1, we can pick the complex structure of $M$ to be such that the $3g-5$ independent compressible circles
correspond to disjoint simple closed geodesics that are all very short.   Let $\gamma$ be one of these.   As $\gamma$ is compressible, it is the
boundary of a disc $P\subset X$.  In appendix \ref{details}, we show that one can choose $P$ to be a totally geodesic plane (thus, a copy of
$\AdS_2$) embedded in $X$.   This makes it possible to perform a simple ``surgery'' on $X$.   We cut $X$ along $P$.   Along each of the resulting
boundaries, we glue in a copy of half of $\AdS_3$.   (Concretely, one can cut $\AdS_3$ in half along a geodesic plane $\AdS_2\subset \AdS_3$,
and then glue in the resulting pieces along the cuts in $X$.)   This gives a new hyperbolic manifold $X_1$, which may or may not
be connected.    The effect of the surgery on $M$
is to cut $M$ along $\gamma$ and glue in a disc on each side, giving a new oriented two-manifold $M_1$.  Even if $X_1$ is connected, $M_1$
may not be.   
Components of the conformal boundary of $X$ other than $M$, if there are any, are not affected by the surgery.

If $M_1$ is connected, it has genus $g_1=g-1$.   Of the $3g-5$ independent compressible curves that we started with on $M$, one, namely $\gamma$,
disappeared in the surgery.   The other $3g-4$ independent compressible circles in $M$ are still compressible in $X_1$.   At most 2 of them are no longer
independent in $M_1$ (they are nullhomotopic or homotopic to each other; see fig. \ref{GenusThree}).   Therefore, if $M$ has $3g-5$ independent compressible circles,
then $M_1$ has at least $3g-8=3g_1-5$ such circles.   Now let us assume inductively that Proposition 3.1 is true for a conformal boundary component of
genus $g-1$.   By this inductive hypothesis applied to $M_1$, we learn that $X_1$ is a Schottky manifold.   Given this, an elementary geometric
argument shows that $X$ is also a Schottky manifold, and in particular its conformal boundary is connected.

If $M_1$ is not connected, it is the union of components $M_1'$, $M_2'$ of genera $g_1',g_1''\geq 1$, with $g_1'+g_1''=g$.   First consider the case
$g_1',g_1''\geq 2$.  Let $n_1',n_1''$ be the maximum number of independent compressible circles on $M_1',$ $M_1''$.  We have $n_1'+n_1''\geq 3g-8$
(we started with at least $3g-5$ such circles on $M$; $\gamma$ was lost and at most one compressible circle on $M_1'$ and one on $M_1''$ is no longer independent
after the surgery, leaving at least $3g-8$ of them).   On the other hand, $n_1'\leq 3g'_1-3$, $n_2'\leq 3g_1''-3$, since a Riemann surface
of genus $h$ supports at most $3h-3$ independent circles.    So $n_1'\geq 3g_1'-5$, $n_1''\geq 3g_1''-5$.   Hence by the inductive hypothesis applied
to $M_1'$ and $M_1''$, $X_1$ must consist of disjoint components $X_1'$ and $X_1''$, where $X_1'$ is a Schottky manifold of conformal boundary $M_1'$, and $X_2'$ is
a
Schottky manifold  of conformal boundary $M_2'$.
   Again it follows by an elementary geometric
argument that $X$ is also a Schottky manifold, in particular with connected conformal boundary.

The same conclusion applies if $M_1'$ and/or $M_1''$ has genus 1.   For example, if $M_1'$ has genus $1$, then $M_1''$ has genus $g_1''=g-1$.
In this case, of the original compressible circles on $M$, at most 1 was originally on $M_1'$ and (assuming $g_1''\geq 2$)
 at most 1 which was originally on $M_1''$ is no longer
independent after the surgery.  So there are at least $3g-8=3g_1''-5$ independent compressible circles on $M_1''$, and the inductive hypothesis
applies to $M_1''$ as before.    As for $M_1'$, since it has genus 1, we invoke the fact that a hyperbolic three-manifold whose conformal boundary
has a genus 1 component is one of the $X_{c,d}$, topologically a solid torus.  So $X$ is the disjoint union of two components, one a solid torus
with conformal boundary $M_1'$ and one a Schottky manifold with conformal boundary $M_1''$.   Again it follows that $X$ is also a Schottky manifold.
A similar argument applies in the special case $g=2$, $g_1'=g_1''=1$.  This completes the proof of Proposition 3.1.

The proof of Theorem 3.2 requires more sophisticated tools.   We must prove two statements: (i) if $\gamma$ is not compressible in $X$, then when $\gamma$
is pinched, $V_R(X)$ remains bounded; (ii) if $\gamma$ is compressible, then when $\gamma$ is pinched, $V_R(X)\to-\infty$.

To prove the first statement, we use the fact that the pinching locus is at finite distance in the Weil-Petersson metric on Teichm\"{u}ller space.   Therefore,
to show that $V_R(X)$ remains bounded as one approaches the pinching locus, it suffices to know that the gradient of $V_R(X)$ in the Weil-Petersson metric remains
bounded.   In fact, it is known that the Weil-Petersson gradient of $V_R(X)$ remains bounded as long as no curve in $M$ that is compressible in $X$
becomes short.   Specifically, if $\ell$ is the length of the shortest non-trivial closed curve in $M$ that is compressible in $X$, and $\chi(M)$ is the Euler
characteristic of $M$, then  \cite{bridgeman-canary:renormalized,bridgeman-brock-bromberg:gradient} the gradient satisfies the bound
\be\label{gradientbound}   \| \d V_R\|_{\mathrm{WP}} \leq \frac{3\sqrt{\pi|\chi(\partial_\infty X)|}}{\sqrt 2\tanh^2(\ell/4)}. \ee
For completeness, we provide a proof in appendix \ref{details}.

For the second statement, we need an upper bound on $V_R(X)$.   There is a useful upper bound in terms of the
volume of the convex core of $X$, which we will denote as $V_C(X)$, and the length of the measured bending lamination of the convex core, which we will
denote as $L_m(l)$.   The meaning of these terms is described in appendix \ref{details}.
In terms of these quantities, one has a bound on $V_R(X)$:
\be\label{boundvr} V_R(X)\leq V_C(X) -\frac{1}{4}L_m(l) +\frac{\pi \log 2}{2}|\chi(M)|. \ee
(The coefficient of the last term on the right of this equation depends on a choice of normalization in the definition of $V_R$.)    This inequality can be found in
 \cite[Theorem 1.1]{compare} for quasi-Fuchsian manifolds, but the proof extends without change for general $X$; see \cite[Section 3]{bridgeman-canary:renormalized}.
 
 In appendix \ref{details}, we show that  $V_C(X)$ remains bounded when one or more compressible curves is pinched.  One also has a result of
 Bridgeman and Canary (see \cite[Theorem 2]{bridgeman-canary:bounding}, and also \cite[Theorem 4.2]{bridgeman-canary:renormalized}) showing
 that $L_m(l)\to\infty$ when a compressible curve is pinched.   Specifically, there are constants $P,Q$ (one can take $P=74$ and $Q=36$) such that if
 $M$ contains a compressible closed geodesic of length $r<1$, then 
 \be\label{lowerbound}L_m(l)\geq \frac{P}{r}-Q.\ee
 The second part of Theorem 3.2, asserting that $V_R(X)\to-\infty$ when a compressible curve is pinched, follows from
 the bounds (\ref{boundvr}) and (\ref{lowerbound}) along with the fact that $V_C(X)$ remains bounded in this limit.  
 
 Physically, one would expect a more precise result than we have stated so far.   One would expect that when a compressible cycle is pinched,
 the divergence of $V_R$ would precisely reflect the CFT ground state propagating through the cycle in question.    With some more detailed
 arguments, we establish this in Theorem A.15.

\section{Why Do Some Observables Show Ensemble Averaging?}\label{why}

As explained in the introduction, connected amplitudes with disconnected boundary, or CADB amplitudes for short, have been a puzzle since early
days of the AdS/CFT correspondence.  A possible explanation has been that actually, the dual of a specific bulk theory is the average of an ensemble
of boundary theories, rather than a specific boundary theory.  Averaging over an ensemble can readily generate CADB amplitudes.
There is a standard objection to this proposal: in many examples of AdS/CFT
duality, it is believed that all of the parameters that the CFT can depend upon (consistent with its general properties such as the supersymmetry algebra
it satisfies) are known, and the bulk theory depends on all of the same parameters.   So what ensemble could one possibly be averaging over to generate
CADB amplitudes?

In this article, we have attempted to sharpen this puzzle by arguing that a certain important class of observables, namely the ones that can be defined
purely in terms of energies and couplings of states that are below the black hole threshold, does not receive any contributions with disconnected
boundaries and thus is not affected by  ensemble averaging.

If there is no ensemble to average over, and if states below the black hole threshold do not show any sign of ensemble averaging, why is it that  when we compute
observables involving black holes states, the gravitational
path integral appears to give ensemble-averaged answers?  Clearly the answer must involve some essential difference
between fixed energy states and black hole states.

Here we will propose a simple answer to this question, based on two assertions:

\begin{itemize}\item{ Black hole physics is highly chaotic.}
\item{
The Hamiltonian $H_N$ describing black hole states does not have a large $N$ limit, and likewise other CFT observables involving
black hole states, such as the trilinear couplings $\la \O_i\O_j\O_k\ra$, do not have a  large $N$ limit, even in a rather general sense,
as will be further discussed presently.  } \end{itemize}

The first statement is generally accepted, based on a reinterpretation \cite{Chaos} of  older calculations \cite{THD} of the behavior of perturbations
in the field of a black hole.  This statement involves a contrast between black holes and fixed energy states, because in a number of important examples,
the spectrum of fixed energy states is believed to be described by an integrable model \cite{Beisert}, not a system with chaotic behavior.  
The second statement also involves a contrast between black holes and fixed energy states. 
Fixed energy states are the states that we see if we take $N\to\infty$ keeping fixed the excitation energy above the ground state.
AdS/CFT duality implies that the energies and couplings of such states have a large $N$ limit; when the boundary theory is a gauge
theory, this can also be seen via a  classic analysis of Feynman diagrams \cite{Thooft}.  
To reach the black hole region, we take $N\to\infty$ with an excitation energy of order $N^2$ if the boundary theory is a gauge theory (and a different
positive power of $N$ in other cases).
The literature does not contain any proposal concerning a sense
in which the Hamiltonian and other observables of black hole states have a large $N$ limit.   Since the entropy (for black hole states of a fixed temperature)
is also growing as a power of $N$, the dimension of the black hole Hilbert space (at a fixed temperature) 
increases by a vast factor from one large value of $N$ to the next.  For example, in the case that the boundary theory is a gauge theory, since the entropy is
asymptotically  $S=bN^2$, with $b$ of order 1, when one changes $N$ from $10^6$ to $10^6+1$, the dimension $e^S$ of the Hilbert space increases by a vast  factor $e^{2b \times 10^6}$.
 This makes
it unclear in what sense one might hope that the black hole Hilbert space and other
observables would have a large  $N$ limit.   In the somewhat analogous problem of quantum statistical mechanics with the volume $V$ playing the role of $N$,
the standard answer is that the Hilbert space and Hamiltonian do not have a large $V$ limit. 
By contrast, the thermofield double state of a pair of entangled systems does have both a large $V$ \cite{HHW} and large $N$ \cite{MaldaDouble} limit. See
\cite{WittenLecture} for more discussion.  

Most likely,  the black hole Hamiltonian and couplings do not have a large $N$ limit, in the sense that, in general, energies and couplings of black hole
states do not have any regularity for large $N$ beyond what follows from the fact that thermodynamic functions and other averages over the spectrum
depend smoothly on $N$ and that, similarly, certain asymptotic
averages of functions of couplings are also smooth functions of $N$.   Asymptotic formulas for averages of functions of couplings were introduced in \cite{CMM} and have
been studied in a number of more recent papers.

Our proposal is that these differences between black hole states and fixed energy states are the reason that apparent ensemble averaging affects
black hole states and not fixed energy states.
Let $H_N$ be the CFT Hamiltonian  at given $N$, on a sphere $\S^{d-1}$ with round metric. $H_N$ commutes with
a symmetry group $G$ consisting of rotations of $\S^{d-1}$ and possible additional symmetries of the CFT, and so is block diagonal with
blocks labeled by representations of $G$.  Chaos in black hole physics
means that if we restrict to states in a band of energies that is above the black hole threshold, then $H_N$ in each block
is an enormous pseudorandom matrix.  A pseudorandom matrix is a matrix that cannot be distinguished from a truly random matrix by any simple
measurement.
   If it is true that $H_N$ does not have a large $N$ limit above the black hold threshold, this suggests
that in each block, 
  the $H_N$ for neighboring values of $N$ can be viewed as  independent pseudorandom draws from a random matrix ensemble.   (As we explain later, it
seems that this statement is actually subject to  corrections that are exponentially small in $N$, but it can serve as a first approximation.)  The
random matrix ensemble is characterized by specifying the entropy $S$ as a function of the energy $E$ and other conserved charges, so it depends smoothly
on $N$.  

Let us consider CFT observables that can be constructed just in terms of $H_N$.   The most important such observables are the twisted partition functions
$Z_{N,R}(\beta)=\Tr\,e^{-\beta H_N}R$, where $\Re\,{\beta}$ is positive and is small enough that the trace is dominated by black hole states, and $R\in G$.
But the following explanation may be clearer if we think first about an arbitrary observable $\W_N$ that depends only on the pseudorandom matrix
$H_N$.   $\W_N$ may be a ``self-averaging''
function in random matrix theory, meaning that it has almost the same value for almost any draw from the random matrix ensemble.
In that case $\la \W_N\ra$  will be a smooth function of $N$, modulo exponentially small corrections that reflect the fact that even self-averaging
functions of a random matrix differ slightly from draw to draw. 
(These corrections are exponentially small because the size $e^S$ of the random matrix is exponentially large, as observed in \cite{SSS1}.) The  corrections to self-averaging behavior will depend erratically on $N$, since they depend on a pseudorandom draw which is different for each $N$.
 If $\W$ is not self-averaging, it will be an erratic function of $N$.  With presently known
methods, the gravitational path integral always produces a smooth function of $N$, typically by summing over contributions of saddle points corresponding to
classical solutions.  Even when classical solutions are not available, calculations that we know how to perform lead to smooth functions of $N$,
as in \cite{CJ}.

Based on this,  what might be calculable with presently available methods?  If $\W_N$ is self-averaging, we can hope to calculate $\la \W_N\ra$ modulo exponentially
small terms that depend erratically on $N$ and depend on a particular draw from the random matrix ensemble.   
 If $\W_N$ is not self-averaging, we will
not be able to compute any approximation to $\la \W_N\ra$  with presently known methods. 
However, an observable that is not self-averaging might still have a nonzero average value in a random matrix ensemble (see \cite{SSS1} for examples), 
and it might be possible to compute this
from the gravitational path integral.   In that case, the expression for $\la \W_N\ra$ that the gravitational path integral would compute would really be an average
value, averaged over nearby values of $N$.

 Now consider several   observables $\W_{N,k}$, $k=1,\cdots,s$ that are all functions of the pseudorandom matrix
$H_N$.    Whether or not individually they have nonzero
averages, the connected correlation function $\la \W_{N,1}\W_{N,2}\cdots \W_{N,k}\ra_c$ may have a nonzero average in random matrix theory,
in which case we may be able to compute this average from the gravitational path integral.
Let us focus on the special case $\W_{N,k}=\Tr\,e^{-\beta_k H_N} R_k$, for some $\beta_k$, $R_k$.  In this special case, $\W_{N,k}$ is a partition function 
$Z_{N}(\beta_k,R_k)$ on a manifold $\S^1\times \S^{d-1}$.  Here if  $\beta_k$ is real, it is the circumference of $\S^1$; it is also interesting to analytically
continue these observables to complex $\beta$, as in \cite{SSS1}.
The group element $R_k\in G$ determines a holonomy around the $\S^1$ factor; this holonomy
consists of a rotation of $\S^{d-1}$ and/or an internal symmetry.   From what we have just said, the gravitational path integral with known methods
may be able to calculate  an averaged value of the connected correlation function
\be\label{conpar}\la Z_N(\beta_1,R_1)Z _N(\beta_2,R_2)\cdots Z_N(\beta_k,R_k)\ra_c .\ee
How the gravitational path integral would calculate this function, or more precisely an approximation to it with a smooth dependence on $N$, is not
immediately clear just from the hypothesis that the $H_N$ are independent pseudorandom matrices.    But using everything we know
about path integrals and quantum gravity, the obvious hypothesis is that $\la Z_N(\beta_1,R_1)Z _N(\beta_2,R_2)\cdots Z_N(\beta_k,R_k)\ra_c $
should be computed from a path integral with a connected bulk and a boundary that is the disjoint union of $k$ copies of $\S^1\times \S^{d-1}$.

This is a plausible interpretation of CADB amplitudes for the special case that the boundary is a union of copies
of\footnote{Considering this example first made possible  a description in terms of $H_N$ only, which was helpful, because random matrix theory is on a much clearer footing than random CFT, which we require in a more general case.
But unfortunately, this example is actually inconvenient from a different point of view, because $\S^1\times \S^{d-1}$ has positive Ricci scalar.   In any dimension,
the boundary of an asymptotically AdS solution of Einstein's equations, if not connected, does not contain any component of positive Ricci scalar \cite{WittenYau}.  So a bulk
computation of 
the observables in eqn. (\ref{conpar}) has to rely on contributions that are less well understood, perhaps somewhat along the lines of \cite{CJ}. }
 $\S^1\times \S^{d-1}$.  If it is correct, then presumably something similar must be true for CADB amplitudes with the $k$ copies of $\S^1\times \S^{d-1}$ replaced by more
general $d$-manifolds.   The rough idea must be that the CFT at a specific large value of $N$, though actually it is a definite CFT (dependent in some cases on a few known parameters), looks, if one only
has access to asymptotic expansions near $N=\infty$,  like a pseudorandom
solution  of the axioms\footnote{For example, an important axiom that contains much of the content of CFT
 is a quadratic ``crossing'' equation satisfied by the trilinear couplings $\la\O_i\O_j\O_k\ra$. This relation is found
by comparing different ways to analyze a four-point function $\la\O_i\O_j\O_k\O_l\ra$.}  of CFT.   Then one would repeat everything we have said so far with
the assertion that the $H_N$ for different $N$ are independent pseudorandom matrices replaced by the statement that the CFT's for different $N$ are independent
pseudorandom draws from a family of asymptotic  solutions of CFT axioms.  

  Since it is not believed that an ensemble of CFT's with the appropriate properties
actually  exists, the idea here is really that the ensemble of random solutions of CFT axioms from which a given large $N$ CFT appears to be drawn only
exists in an asymptotic sense, for large $N$.   A rough analogy is that in low energy effective field theory, the $S$-matrix of a relativistic quantum field theory appears
to be a special case of a family of unitary, relativistic $S$-matrices that can be obtained by giving arbitrary coefficients to all possible parameters in the low energy
effective action.     It is generally believed that the generic element of this family  exists only as an asymptotic expansion at low energies.

Thus, our proposal can be stated as follows.  The CFT's that govern black hole states for different large values of $N$ look, in simple measurements, like (nearly)
independent
pseudorandom draws from a ``swampland'' of effective CFT's that are defined asymptotically for large $N$ and cannot be completed to true theories at integer
values of $N$.   This CFT ``swampland'' would be analogous to the usual ``swampland'' of low energy
effective field theories that are believed not to have ultraviolet completions \cite{Vafa}.   The gravitational path integral, with known methods, calculates averages
over the pseudorandom CFT's with neighboring values of $N$.

Finally, we should point out that in the context of AdS/CFT duality, it is not true that the CFT's for different $N$ are truly independent above the black hole
threshold.   That is because (in known examples)  the theories
with different values of $N$ are unified in string/M-theory and are connected by domain walls.   We will illustrate this point with 
 a simple example that generalizes the Fuchsian manifold that was
introduced in eqn. (\ref{expmet}).   Let $M$ be a compact hyperbolic $d$-manifold with metric $\d\Omega^2$ 
and set $X=M\times \rR$.   On $X$ there is a complete hyperbolic metric
\be\label{compmet}\d s^2=\d t^2+\cosh^2 t\, \d\Omega^2.\ee
The conformal boundary of $X$ consists of two copies of $M$, at $t=\pm\infty$.   Let $U$ be the submanifold of $X$ defined by $t=0$.   Then $U$ is a minimal
submanifold,\footnote{In $D=3,$ $U$ is the convex core of $X$; see footnote \ref{convcore}.}  so there is a classical solution in which a brane is placed on $U$.   If this brane is of the appropriate type, the integer that characterizes
the CFT (or one of those integers in the case of a CFT that depends on multiple integers) will jump from $N$ to $N+1$ in crossing $U$.   Thus a path
integral on $X$ in the presence of this brane generates a connected correlation function between partition functions with different values of $N$ on the 
same manifold $M$:
\be\label{conparo}\la Z_N(M) Z_{N+1}(M)\ra_c\not=0. \ee
So the pseudorandom matrices or CFT's for different values of $N$ are not truly independent.   However, they are nearly independent, in the sense that
\be\label{conpatwo} \left|\la Z_N(M) Z_{N+1}(M)\ra_c\right|   \ll \sqrt{ \left|\la Z_N(M) Z_{N}(M)\ra_c\right|   \left|\la Z_{N+1}(M) Z_{N+1}(M)\ra_c\right|  }, \ee
because the brane action contributes to the left hand side and not to the right hand side.   Hopefully this is enough to justify the explanation of CADB amplitudes based on
pseudorandomness.  Still, 
the existence of correlations between the theories for different values of $N$ seems to mean that the Hamiltonians and CFT's of different $N$ are not
truly independent pseudorandom objects.   Perhaps corrections involving branes  lead to exponentially small departures from what one would expect
based on independent draws from a random ensemble.

One may summarize what we have said as a proposal that ensemble averaging in gravity is averaging over nearby values of $N$ to produce smooth approximations
that can be computed by the gravitational path integral with known methods.   That obviously leaves the question of what kind of path integral or what new method is needed,
at least in principle, to describe the non-smooth contributions. There have been several papers aiming to find simple models of how this
can work  \cite{SSSrecent,Baur}.

\vskip1cm
 \noindent {\it {Acknowledgements}}   Research of JMS supported in part by FNR Grant O20/14766753.
 Research of EW supported in part by NSF Grant PHY-1911298.     
 JMS thanks Ian Agol, Martin Bridgeman and Ken Bromberg for useful remarks and references.
   EW thanks L. Takhtajan and Jinsung Park for explanations about the renormalized volume, and K. Krasnov for  helpful
 advice. 

\appendix

\newcommand{\DD}{{\mathbb D}}
\newcommand{\HH}{{\mathbb H}}
\newcommand{\N}{{\mathbb N}}
\newcommand{\cT}{{\mathcal T}}
\newcommand{\CP}{{\mathbb{CP}}}
\newcommand{\R}{{\mathbb R}}
\newcommand{\cS}{{\mathcal S}}
\newcommand{\ext}{{\mathrm{Ext}}}
\renewcommand{\mod}{\mathrm{Mod}}

\section{Mathematical details}
\label{details}

This appendix contains detailed proofs of Proposition 3.1 and of Theorem 3.2, which were already explained heuristically in Section 3, as well as Theorem \ref{tm:3g-6}, which slightly improves on Proposition 3.1. In addition, we provide two results which help better understand the properties of the renormalized volume under pinching of compressible curves.\footnote{Those to results were not contained in the first arxiv version.}
\begin{itemize}
\item Theorem \ref{tm:compare2}, which shows that the renormalized volume associated to the hyperbolic metric at infinity, denoted by $V_R$ here, is within a bounded constant, depending only on the topology of the boundary, from the renormalized volume associated to the Thurston metric at infinity, denoted here by $V'_R$. Note that $V'_R$ is equal to the volume of the convex core minus one fourth of the length of the measured bending lamination on its boundary.
\item Theorem \ref{tm:asymptotics}, which gives the first term in the asymptotic development of the renormalized volume when a compressible curve is pinched. 
\end{itemize}

The arguments are quite elementary but based on recent developments in the study of the renormalized volume of hyperbolic manifolds, which has recently been a focus of some interest among hyperbolic geometers. The renormalized volume was found to have close relations to topics of interest in geometry, and to be a useful or promising tool for well-established mathematical questions. We list here some of those developments.

A first motivation stemmed from the identification in \cite{Krasnov:2000zq,Krasnov:2001cu} between the renormalized volume of (some) hyperbolic manifolds and the Liouville functional studied for instance in \cite{TZ-schottky,takhtajan-teo}.

Another connection was made in \cite{volume,review,compare} between the renormalized volume and the volume of the convex core of convex cocompact hyperbolic manifolds. This relationship was then used for instance in  \cite{kojima-mcshane}, to relate the entropy of pseudo-Anosov diffeomorphisms to their hyperbolic volume of their mapping torus, in  \cite{loustau:minimal,cp} to study the symplectic structure on moduli spaces of quasi-Fuchsian manifolds, and in \cite{brock-bromberg:inflexibility2} to study the metric geometry of moduli space (such as its inradius or systole). In addition, geometric properties of the renormalized volume were investigated, such as its convexity at the critical points \cite{moroianu-convexity,vargas-pallete-local} and continuity under geometric limits \cite{vargas-pallete-continuity,pallete:additive}.

It was proved in \cite{ciobotaru-moroianu} that the renormalized volume of almost-Fuchsian manifolds (quasi-Fuchsian manifolds containing a closed  minimal surface with principal curvatures less than 1) is non-negative, a result that was then extended to quasi-Fuchsian hyperbolic manifolds \cite{bridgeman-bromberg-pallete} and more generally convex co-compact manifolds with incompressible boundary \cite{bridgeman-brock-bromberg}. In contrast, the renormalized volume of hyperbolic manifolds with compressible boundary can be negative -- this remark, which plays a key role here, already appeared e.g. in \cite{bridgeman-canary:renormalized,pallete:schottky}.

A particularly active current direction of research concerns the Weil-Petersson gradient flow of the renormalized volume \cite{bridgeman-brock-bromberg,bridgeman-brock-bromberg:gradient,bridgeman-bromberg-pallete}, considered as a tool to understand the structure of 3-dimensional hyperbolic manifolds.

The properties of the renormalized volume for Schottky manifolds are considered specifically in \cite{pallete:schottky}, in view of the comparison of volumes of quasi-Fuchsian and Schottky manifolds with a given conformal boundary.





Since this section is geared towards more mathematical arguments, we use a slightly different notation than in the previous sections. We will always consider the hyperbolic space $\HH^3$ of constant sectional curvature $-1$, which is equivalent to setting $\ell=1$.

\subsection{Convex co-compact hyperbolic manifolds}
\label{ssc:convex}

Before entering the arguments, it is useful to clarify some definitions.

We consider here a complete hyperbolic structure $g$ on an oriented 3-dimensional manifold $X$, which will always be the interior of a compact manifold with boundary. Such a hyperbolic structure is the quotient of the 3-dimensional hyperbolic space $\HH^3$ by $\rho(\pi_1X)$, where $\rho:\pi_1X\to PSL(2,\C)$ is the holonomy representation of $(X,g)$.

The boundary at infinity of $\HH^3$ can be identified with $\CP^1$, and it is tempting to consider $\rho$ as an action of $\pi_1X$ on $\partial_\infty \HH^3=\CP^1$. However, this action on $\CP^1$ is not properly discontinuous, so that one cannot take the quotient. To avoid this issue, one needs to ``remove'' from $\CP^1$ the {\em limit set} $\Lambda_\rho$ of $\rho$, defined as the intersection with $\partial_\infty\HH^3$ of the closure in $\HH^3\cup \partial_\infty\HH^3$ of the orbit $\rho(\pi_1X)(x)$ of any point $x\in \HH^3$. It turns out (see Section \ref{ssc:infinity}) that $\rho$ acts properly discontinuously on $\CP^1\setminus \Lambda_\rho$.

We say that the subgroup $\rho(\pi_1X)$ is {\em elementary} if its limit set has at most 2 points.

A hyperbolic manifold is {\em convex co-compact} if
\begin{itemize}
\item its holonomy representation acts co-compactly (i.e. with compact quotient) on a convex domain in $\HH^3$, and
\item the image of its fundamental group in $PSL(2,\C)$ is non-elementary.
\end{itemize}
In other terms, it is the quotient of $\HH^3$ by a non-elementary subgroup of $PSL(2,\C)$, which contains a non-empty compact geodesically convex subset.\footnote{The term ``convex co-compact'' is perhaps a bit misleading. What can properly be called convex co-compact is rather the holonomy representation $\rho:\pi_1X\to PSL(2,\C)$, since it acts on a convex subset (the convex hull of $\Lambda_\rho$ in $\HH^3$) with compact quotient.}

\begin{definition}
  Let $X$ be a hyperbolic manifold. A subset $K\subset X$ is {\em geodesically convex} if any geodesic segment of $X$ with endpoints in $K$ is contained in $K$.
\end{definition}

Note that geodesic convexity is a strong property, for instance a small ball in a complete hyperbolic manifold with non-trivial fundamental group is not geodesically convex. In fact,  if $K$ is a non-empty geodesically convex subset of $X$, then the inclusion of $K$ in $X$ is a homotopy equivalence, see Section \ref{ssc:A-convex}. 

Here we will use the equivalent definition of a convex co-compact manifold, which is more convenient for the proofs.

\begin{definition} \label{df:ccv}
    A convex co-compact hyperbolic structure on a manifold $X$ is a complete hyperbolic structures for which $X$ contains a non-empty, {\em compact}, geodesically convex subset $K$, and such that $X$ is not topologically a ball or a solid torus. 
\end{definition}

We exclude from the definition the case where $X$ is a ball or a torus, which correspond to elementary group actions. Therefore, a complete hyperbolic manifold which contains a compact, non-empty, geodesically convex subset can be either $\HH^3$, a solid torus, or a convex co-compact manifold as defined here.

For a hyperbolic manifold, being convex co-compact is equivalent to being conformally compact, that is, to having a Riemannian metric that can be written as $g=\bar g/\rho^2$, where $\bar g$ is a Riemannian metric which is smooth on $\bar X$ up to the boundary, while $\rho:\bar X\to \R_{\geq 0}$ is a smooth function that vanishes on the boundary, with $\| d\rho\|_{\bar g}=1$ on $\partial X$. Indeed:
\begin{itemize}
\item If $(X,g)$ is conformally compact, a direct computation shows that the surfaces
  $$ S_\epsilon = \{ x\in X~|~\rho(x)=\epsilon\} $$
  are locally convex for $\epsilon>0$ small enough. This simplies that the (compact) set
  $$ X_\epsilon = \{ x\in X~|~\rho(x)\geq \epsilon\} $$
  is geodesically convex for $\epsilon>0$ small enough. Indeed, a geodesic segment with endpoints in $X_\epsilon$ must stay in $X_\epsilon$ since otherwise, at the point where $\rho$ achieves its minimum $\rho_{min}$, it would need to be tangent to $S_{\rho_{min}}$ on the convex side, a contradiction.
\item Conversely, if $X$ is convex co-compact, it contains a geodesically convex subset $K$ which is compact. Replacing if necessary $K$ by an $r$-neighborhood and smoothing its boundary, we can assume that $K$ has smooth boundary. If $r:X\to \R_{\geq 0}$ is defined as the distance to $K$, the function $\rho=e^{-r}$ is a defining function and $(X,g)$ is conformally compact.
  
\end{itemize}
\subsection{The complex structure at infinity}
\label{ssc:infinity}

The set $\Omega_\rho=\partial_\infty \HH^3\setminus \Lambda_\rho$ is called the {\em discontinuity domain} of $\rho$. Since $\rho$ acts by hyperbolic isometries on $\HH^3$, it acts by complex transformations on $\Omega_\rho$, and it can be proved that this action is properly discontinuous, see \cite[Sections 8.1 and 8.2]{thurston-notes}. The quotient $\partial_\infty X = \Omega_\rho/\rho(\pi_1X)$ is therefore equipped with a complex structures, which will be denoted by $c$ here.

By a series of results of Ahlfors, Bers, Kra, Marden, Maskit, Sullivan and Thurston, a convex co-compact hyperbolic metric on $X$ is  uniquely determined by $c$, considered as a point in the Teichm\"uller space $\cT_{\partial X}$ of $\partial X$. If $X$ has incompressible boundary, then this map from $\cT_{\partial X}$ to the moduli space of convex co-compact hyperbolic metrics is one-to-one. However, if $X$ has compressible boundary, two points in $\cT_{\partial X}$ can determine the same convex co-compact structure on $X$. This happens when one is the image of the other by an (isotopy class of) homeomorphism which extends over the manifold -- for instance, a homeomorphism corresponding to a Dehn twist along a compressible simple closed curve (a curve in $\partial X$ which bounds a disk in $X$).

As a consequence, the space of convex co-compact hyperbolic structures on $X$ is parameterized in $\cT_{\partial X}/\Gamma$, where $\Gamma$ is the group of isotopy classes of $X\cup \partial X$ which are homotopic to the identity, see \cite[Section 3]{canary:pushing}. 

Note that $\partial_\infty X = \Omega_\rho/\rho(\pi_1X)$ is equipped with more than a complex structure: each point has a neighborhood that can be identified with a domain in $\CP^1$, and this identification is well-defined up to elements of $PSL(2,\C)$.\footnote{This can be formalized as the existence on $\partial_\infty X$ of a {\em complex projective structure}, but this point of view will not be necessary here.} The existence of those local charts in $\CP^1$ will be relevant in Section \ref{ssc:schwarzian}.

\subsection{Measured laminations on surfaces}
\label{ssc:A-laminations}

Measured laminations play a significant role in the arguments below, so we provide here a brief introduction to their definition and key properties. Measured laminations occur in the next section when describing the geometric structure on the boundary of the convex core of a convex co-compact hyperbolic manifold.

Let $M$ be a closed surface, equipped with a hyperbolic metric $h$ -- one can consider more generally complete hyperbolic surfaces of finite area (or even, with some adaptations, of infinite volume). A geodesic lamination is then defined as a closed subset of $M$ which is a disjoint union of complete geodesics. A {\em measured} geodesic lamination is a geodesic lamination equipped with a transverse measure, that is, each transverse curve is equipped with a measure, and this measure does not change when the curve is moved while keeping the intersection with the lamination transverse, see \cite[Section 10]{thurston:minimal}.

The simplest case of geodesic laminations is the disjoint union of a finite family $c_1,\cdots, c_n$ of disjoint closed geodesics. A transverse measure is then defined simply as a positive weight $w_i$ associated to each geodesic $c_i$, yielding a {\em weighted multicurve}. There is a natural topology on the space of weighted multicurves, where two weighted multicurves are close if they have a similar intersection with any transverse closed curve on $M$. The space of measured geodesic laminations can be defined as the completion of the space of weighted multicurves for this topology. It follows that weighted multicurves (and in fact even weighted closed geodesics) are dense in the space of measured geodesic laminations. 

However, ``generic'' geodesic laminations can be more complex than weighted multicurves. While their support has Hausdorff dimension equal to 1 \cite{birman-series}, a short transverse segment might have an uncountable set of intersections, none of which has an atomic weight. 

Given a hyperbolic metric $m$ on $M$ and a measured lamination $l$, Thurston (see e.g. \cite[Section 2]{thurston:hyperbolicII}) proved that one can define the {\em hyperbolic length} of $l$ for $m$, denoted here by $L_m(l)$. It is defined as the limit of the hyperbolic lengths of a sequence of weighted multicurves $((c^n_i,w^n_i)_{i=1, \cdots, k_n}$ converging to $l$, where the hyperbolic length of a weighted multicurve $(w_i, c_i)_{i=1,\cdots, k}$ is defined as
$$ L_m((w_i, c_i)_{i=1,\cdots, k})=\sum_{i=1}^k w_iL_m(c_i)~. $$

The notion of measured lamination does not in fact require a hyperbolic metric, and can be defined ``topologically''. Any measured lamination on a closed surface $M$ can then be realized uniquely as a geodesic measured lamination for any hyperbolic metric $m$ on $M$, much like any closed curve and be realized uniquely as a geodesic.

\subsection{The convex core and the geometry of its boundary}

The arguments presented here rely heavily on the relations between the renormalized volume and the volume of the {\em convex core} of convex co-compact hyperbolic manifolds, as well as on the geometry of the boundary of this convex core.

It follows from the definition that, in a hyperbolic manifold, the intersection of two closed, geodesically convex subsets is still geodesically convex. In addition it can be shown (see Section \ref{ssc:A-convex}) that, if the manifold has non-trivial topology, then the intersection is non-empty. As a consequence, any convex co-compact hyperbolic manifold contains a smallest non-empty, closed, geodesically convex subset, called its convex core, which is compact. We denote it here by $C(X)$.

The convex core has finite volume, denoted here by $V_C(X)$. This volume is closely related to the renormalized volume -- when the boundary is incompressible, there is a bound on the difference between the two, which only depends on the topology of the boundary, see Section \ref{ssc:volumes} below.

The geometry of the boundary of the convex core was analysed by Thurston (see \cite[Section 8.5]{thurston-notes}). Since the convex core is a minimal geodesically convex subset, its boundary cannot have extremal points, so it is a locally convex {\em pleated surface} -- it is the union of a finite set of totally geodesic ideal triangles intersecting along their boundary, and as a consequence the induced metric on the boundary is hyperbolic (of constant curvature $-1$). We will denote this hyperbolic metric on $\partial C(X)$ by $m$. (Note that the ideal triangles mentioned here generally ``wrap'' around the boundary of $C(X)$, so they appear more clearly in the universal cover of $X$, where their vertices can be identified as points of the limit set of $X$.)

However the pleating locus -- the set of points which do not have a neighborhood which is totally geodesic -- can be quite complicated. It is a geodesic lamination, as seen in Section \ref{ssc:A-laminations}. Moreover the pleating locus is equipped with a transverse measure, which records the amount of ``bending'' along each geodesic. In the simplest cases where the bending locus is a disjoint union of closed geodesics, each of the closed geodesics is equipped with a positive weight, which records the exterior angle of the boundary of the convex core along this ``edge''.

In generic examples, no geodesic in the pleating locus has an atomic weight. However simple closed geodesics equipped with a positive weight are dense, in a suitable topology (see Section \ref{ssc:A-laminations}) in the space of measured laminations. Heuristically, it is therefore often sufficient to think of (arbitrarily long and complicated) weighted simple closed curves, rather than of generic measured laminations.

Bonahon and Otal \cite{BO} gave a complete description of the measured laminations that can arise in this manner on  the boundary $\partial X$. The conditions are that each closed curve has weight at most $\pi$, that for any essential annulus $A$ or M\"obius band in $X$ (with boundary in $\partial X$), $\partial A$ has positive intersection with $l$, and that for any essential disk $D$ (with boundary in $\partial X$), $\partial D$ has intersection larger than $2\pi$ with $l$.

\subsection{Relations between the boundary of the convex core and the boundary at infinity}
\label{ssc:boundary}

Given a convex co-compact hyperbolic manifold $X$, let $N\partial C(X)$ be the unit normal bundle of $\partial C(X)$, that is, the set of outwards-pointing unit vectors normal to a support plane of $C(X)$ at a boundary point. For each $n\in N\partial C(X)$, let $g_n:\R_{\geq 0}\to X$ be the geodesic ray determined by the initial vector $n$. It can be proved (see Section \ref{ssc:A-convex}) that the images of the $g_n$ are all disjoint (and do not intersect $C(X)$ except at their origin). As a consequence, the map sending $n$ to the endpoint at infinity of $g_n$ defines a homeomorphism between $N\partial C(X)$ and $\partial_\infty X$.

Moreover, when $X$ is not Fuchsian\footnote{A Fuchsian manifold is the quotient of $\HH^3$ by a surface group acting properly discontinously and cocompactly on a totally geodesic plane. So the convex core of a Fuchsian manifold is a totally geodesic surface, while for any other convex co-compact manifold it is a 3-dimensional domain with positive volume.   See also Section 2.3, and specifically footnote 9.}, this homeomorphism can be deformed, by ``smoothing out'' the bending locus, to a (non-canonical) homeomorphism between $\partial C(X)$ and $\partial_\infty X$. (If $X$ is Fuchsian, the convex core is a smooth surface, and no smoothing is necessary.) We will often use this homeomorphism implicitly below, and identify the two surfaces.  (If $X$ is Fuchsian, then $\partial_\infty X$ is homeomorphic to the disjoint union of two copies of $C(X)$, which in this case is a totally geodesic surface, so that $N\partial C(X)$ is a double cover of $C(X)$.)  

The boundary at infinity $\partial_\infty X$ is equipped with a conformal structure, and therefore with a hyperbolic metric provided by the Poincar\'e-Riemann uniformization theorem -- it is called the Poincar\'e metric, and denoted here by $h$. It can be compared to the induced metric $m$ on the boundary of the convex core.
\begin{itemize}
\item A simple closed curve is short for $h$ if and only if it is short in $m$. Specifically, for all $\epsilon>0$ small enough, there exists $\epsilon'>0$ such that if a simple closed curve $\gamma$ has length less than $\epsilon'$ for $h$, then it has length less than $\epsilon$ for $m$, and conversely (with a different value of $\epsilon'$), see \cite{sugawa:uniform,canary:conformal}.
\item Curves which do not enter the ``thin'' part of $\partial X$ for either $m$ or $l$ (the subset composed of points where the injectivity radius is smaller than a fixed constant $\epsilon_0$ -- for $\epsilon_0$ small enough, this thin part is the disjoint union of long and ``thin'' annuli, each associated to a simple closed geodesic of length less than $\epsilon_0$) have lengths for $m$ and for $l$ which are comparable, up to bounded multiplicative constants, see \cite{bridgeman-canary}.
\end{itemize}

\subsection{The maximal set of contractible curves in the boundary}
\label{ssc:maximal}

The proof of Proposition 3.1, presented in Section 3, is based on a surgery that can be applied on convex co-compact hyperbolic structures when the length of a compressible curve (in the Poincar\'e metric at infinity) is sufficiently short. Similar ideas will be used again below in Section \ref{ssc:convexcore} when considering the limit of the volume of the convex core when a finite set of disjoint closed compressible curves in the boundary is pinched.

We consider a convex co-compact hyperbolic manifold $X$, a connected component $M$ of $\partial C(X)$, and a simple closed curve $\gamma$ in $M$ which is compressible, i.e., bounds a disk in $X$. We first deform $X$ so that the geodesic representative of $\gamma$ is pinched to have length less than $\epsilon$ (for a value of $\epsilon$ that will be made more precise below) for the induced metric $m$ on $\partial C(X)$. Then the geodesic representative $\gamma_0$ of $\gamma$ in $(M,m)$ is the ``center'' of a long collar, of width $w(\epsilon)$ arbitrary large if $\epsilon$ is small.

Let $\Delta$ be a complete geodesic in $C(X)$ intersecting the disk $D$ in $C(X)$ bounded by $\gamma_0$, and let $P$ be a totally geodesic plane in $X$ orthogonal to $\Delta$ and intersecting $\gamma_0$. (Note that $P$ is not entirely determined by those conditions, since $\gamma_0$ is not necessarily contained in a plane orthogonal to $\Delta$.)

The following elementary lemma is used in Section 3. 

\begin{lemma} \label{lm:embedded}
  $P$ is embedded in $X$, that is, it has no self-intersection.
\end{lemma}

\begin{proof}
    Assume the opposite, it would mean that there are two distinct lifts $P_1$ and $P_2$ of $P$ to $\tilde X\simeq\HH^3$ which intersect at a point $x$. Let $\gamma_1$ and $\gamma_2$ be the corresponding lifts of $\gamma_0$, and let $x_1\in \gamma_1\cap P_1$, $x_2\in \gamma_2\cap P_2$. Let $\alpha_1$ be the geodesic segment connecting $x_1$ to $x$ (which is in $P_1$) and let $\alpha_2$ be the geodesic segment connecting $x_2$ to $x$. Finally, let $\beta$ be the geodesic segment connecting $x_1$ to $x_2$.  

  Since $C(X)$ is geodesically convex, the segment $\beta$ is contained in its universal cover $\widetilde{C(X)}$, and is therefore almost orthogonal (for $\epsilon$ small enough) to both $P_1$ and $P_2$, and also to $\alpha_1$ and $\alpha_2$. So if we call $\theta_1$ (resp. $\theta_2$) the angle between $\beta$ and $\alpha_1$ (resp. $\alpha_2$) then both are close to $\pi/2$ as $\epsilon\to 0$. Moreover, $\beta$ has length at least $2w(\epsilon)$, which goes to infinity as $\epsilon\to 0$. However, a standard hyperbolic triangle formula, applied to the triangle with edges $\alpha_1, \alpha_2$ and $\beta$, ensures that
  $$ \cos(\theta) = - \cos(\theta_1)\cos(\theta_2)+\cosh(L(\beta))\sin(\theta_1)\sin(\theta_2)~, $$
  where $\theta$ is the angle between $\alpha_1$ and $\alpha_2$ at $x$.
  
  For $\epsilon$ small enough, this cannot hold since the two sine terms on the right are close to $1$ and $\cosh(L(\beta))$ is large. 
\end{proof}

\begin{figure}[h]
  \centering
 \includegraphics[width = 5.0cm]{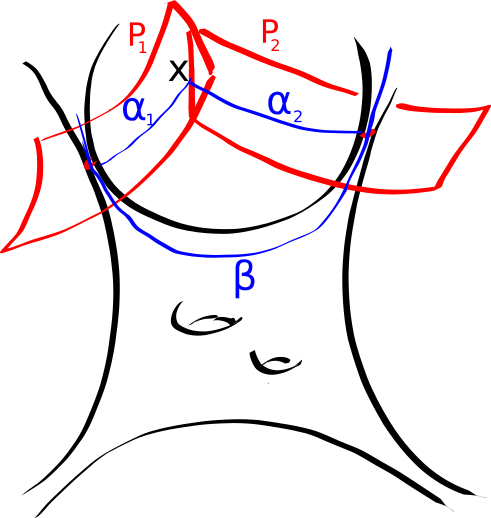}
 \caption{Planes almost orthogonal to thin tubes cannot intersect.}
 \end{figure}

Note that the same argument shows that if we consider two totally geodesic planes $P$ and $P'$ corresponding to two short contractible curves $\gamma$ and $\gamma'$, then $P$ and $P'$ are disjoint -- this is used in Section \ref{ssc:convexcore}.

We are now equipped to prove the following small improvement of Proposition 3.1. The proof repeats the proof of Proposition 3.1 given in Section 3 with more attention given to surfaces of genus $2$.

\begin{theorem} \label{tm:3g-6}
  Let $X$ be a convex co-compact hyperbolic manifold, and let $M$ be a boundary component of $X$. If $M$ has genus $2$ and contains at least one compressible curve, or if $M$ has genus $g\geq 3$ and contains at least $3g-6$ disjoint, non-homotopic compressible curves, then $X$ is a Schottky manifold and $\partial X=M$.
\end{theorem}

\begin{proof}
  The first step is to analyse more carefully what happens to a family of independent curves when a surface $M$ is cut along one of them, say $c$. Different cases can be distinguished, see Figure \ref{GenusThree}.
  \begin{enumerate}
  \item $c$ separates a genus $1$ surface $M_1$ from a genus $g'$ surface $M_{g'}, g'\geq 2$. Then at most one pair of curves on $M$ can become homotopic on $M_{g'}$.
  \item $c$ separates a surface $M_{g'}$ from a surface $M_{g''}$, with $g'\geq 2, g''\geq 2$. Then at most one pair of non-homotopic curves {\em on each side of $c$} can become homotopic.
  \item $c$ is non-separating. Then at most one curve can become homotopically trivial, while at most one pair of non-homotopic curves can be homotopic in the complement of $M$.
  \end{enumerate}

  The proof proceeds by induction on the genus of $M$, as for the proof of Proposition 3.1. The statement is already known, from Section 3, when $M$ has genus $1$ or genus $2$. We assume that it is true for genus at most $g-1$, and consider a boundary component $M$ of genus $g$. We assume that $M$ contains $3g-6$ independent compressible curves, and choose one of those curves, say $c$. We now consider different cases.
  \begin{itemize}
  \item If $c$ splits $M$ into two surfaces $M'$, $M''$ of genus $g'\geq 3$ and $g''\geq 3$. The argument is then exactly as seen in the proof of Proposition 3.1: $g'+g''=g$, the number of remaining independent compressible curves is at least $3g-9=3(g'+g'')-9$, but $M'$ can have at most $3g'-3$ while $M''$ can have at most $3g''-3$. It follows that $M'$ has at least $3g'-6$ independent compressible curves, while $M''$ has at least $3g''-6$. So both are boundaries of Schottky manifolds by the induction hypothesis, and the surgery done on the interiors means that $M$ also bounds a Schottky manifold.
  \item Similarly, the same argument as in the proof of Proposition 3.1 can be applied if $c$ is non-separating and $M$ has genus at least $4$. 
  \item Suppose now that $c$ splits $M$ into $M'$ of genus $g'=2$ and $M''$ of genus $g''\geq 3$. If $M'$ contains at least one compressible curve, then the same argument as before shows that $M''$ contains at least $3g''-6$ independent compressible curves, so by induction $M''$ is the boundary of a Schottky manifold, and therefore $M$ is the boundary of a Schottky manifold.

    If $M'$ contained no compressible curve, then $M''$ would have at least $3g-8$ independent compressible curves, since cutting along $c$ can only ``destroy'' $c$, plus one compressible curve on the side of $M''$ (two compressible curves might become homotopic). But $3g-8=3g''-2$, a contradiction because $M''$ cannot have more than $3g''-3$ independent compressible curves.
  \item If $M$ has genus $4$ and at least $6$ independent compressible curves, and is cut by $c$ into two surfaces of genus $2$, the same argument shows that there must be at least one compressible curve remaining on each side. Otherwise at most 2 curves would be ``lost'' ($c$ plus one pair of compressible curves becoming homotopic) so one side would need to have 4 independent compressible curves, a contradiction.
    \item Finally suppose that $M$ has genus $3$ and at least $3$ compressible curves, and is cut along a non-separating compressible curve $c$. One obtains a surface of genus $2$ containing at least one compressible curve, because at most one compressible curve could become nul-homotopic. So the surface obtained after cutting along $c$ is the boundary of a Schottky manifold, and therefore $X$  is Schottky. 
    \end{itemize}
  \end{proof}

\subsection{A variational formula for the renormalized volume}
\label{ssc:schwarzian}

Although convex co-compact hyperbolic manifolds have infinite volume, one can define a {\em renormalized volume}, see \cite{graham-witten}. We will need a variational formula for the renormalized volume, see \cite{TZ-schottky,takhtajan-teo,volume} or \cite[Corollary 3.11]{compare}. To state it, note that $\partial_\infty X$ is equipped with a holomorphic quadratic differential $q$, which can be defined as follows. Let $\partial_iX$ be a connected component of $\partial_\infty X$. Since $\partial_i X$ is a closed surface of genus at least $2$, equipped with a complex structure (the restriction of $c$ to $\partial_i X$), there is by the Poincar\'e-Riemann uniformization theorem a holomorphic map $f_i:\DD \to \partial_iX$, where $\DD$ is the unit disk in $\C$.

Note that $\partial_\infty X$ can be locally identified with $\CP^1$ (or $\C$, using the holomorphic identification of $\C$ to $\CP^1\setminus \{\infty\}$), and this local identification is well-defined up to left composition by an element of $PSL(2,\C)$. This makes it possible to consider the Schwarzian derivative of $f_i$. Recall that given a holomorphic map $f:\Omega\to \C$, where $\Omega\subset \C$, its Schwarzian derivative is defined as:
$$ \cS(f) = \left(\left(\frac{f''}{f'}\right)' - \frac 12\left(\frac{f''}{f'}\right)^2\right)dz^2~. $$
The following properties are relevant:
\begin{itemize}
\item $\cS(f)=0$ if and only if $f$ corresponds to an element of $PSL(2,\C)$,
\item $\cS(f\circ g)=\cS(f)+f^*\cS(g)$. 
\end{itemize}
It follows from those two properties that the Schwarzian derivative of $f$ is invariant under composition of $f$ on the left by an element of $PSL(2,\C)$. As a consequence, the Schwarzian derivative of $f_i$ is well-defined -- even if the identification of $\partial_\infty X$ with $\CP^1$ is only local. 
We define
$$ q = f_{i*}\cS(f_i)~. $$

By construction, $q$ is a holomorphic quadratic differential on $\partial_i X$.
\footnote{There is another way to introduce $q$, in relation to the second term of the asymptotic development of the metric near infinity when the metric at infinity (corresponding to the first term) has constant curvature. The real part of $q$ is then minus the traceless part of this second term, see e.g. \cite[Lemma 8.3]{volume}. The holomorphic quadratic differential $q$ already appears in this form in \cite{HS}. We will not need this different point of view here.}

\begin{lemma} \label{lm:gradient}
  Let $(c_t)_{t\in [0,1)}$ be a one-parameter family of complex structures in $\cT_{\partial X}/\Gamma$. Then
  $$ \frac{dV_R(c_t)}{dt} = \mathrm{Re}\left(\left\langle q, \frac{dc_t}{dt}\right\rangle\right)~, $$
where $\langle ,\rangle$ denotes the natural duality product between holomorphic quadratic differentials and Beltrami differentials.
\end{lemma}

\subsection{The renormalized volume and the volume of the convex core}
\label{ssc:volumes}

There are close relations between the $V_C(X)$ and $V_R(X)$, given in particular by the following lemma.

\begin{lemma} \label{lm:compare}
  For any convex co-compact hyperbolic manifold $X$,
  $$ V_C(X)-\frac 14 L_m(l)\ - C(\partial X)\leq V_R(X)\leq V_C(X) - \frac 14L_m(l) + \frac{\pi\log(2)}{2}|\chi(\partial X)|~. $$
  where $C(\partial X)$ is a constant depending only on the topology of $\partial X$. 
\end{lemma}

(Note that the additive constant on the right of the equation depends on a choice of normalization in the definition of $V_R$.) The inequality on the right can be found as \cite[Theorem 1.1]{compare} for quasi-Fuchsian manifolds, but the proof extends without change to convex co-compact hyperbolic manifolds, see \cite[Section 3]{bridgeman-canary:renormalized}.  The inequality on the left is Theorem \ref{tm:compare2} below.

There is also a lower bound on renormalized volume, in terms of the volume of the convex core, for convex co-compact manifolds with incompressible boundary, see \cite{compare,bridgeman-canary:renormalized}. For convex co-compact manifolds with compressible boundary, the constant depends on the injectivity radius of the boundary. 

\subsection{A bound on the Weil-Petersson gradient}

We recall here a bound on the Weil-Petersson gradient of $V_R$ when no compressible curve is short. Similar estimates can be found in \cite{kra-maskit:remarks,bridgeman-brock-bromberg}. We provide a proof for completeness. A more precise analysis of the convergence of the geometric structure when an incompressible curve is pinched can be found in \cite{guillarmou-moroianu-rochon}.

\begin{lemma} \label{lm:WP}
  Let $X$ be a convex co-compact hyperbolic manifold, such that the length for the Poincar\'e metric of any non-trivial simple closed curve in $\partial_\infty X$ compressible in $X$ is at least $l$. Then the Weil-Petersson gradient of $V_R$ on $\cT_{\partial X}/\Gamma$ is bounded by
  $$ \| dV_R\|_{WP} \leq \frac{3\sqrt{\pi|\chi(\partial_\infty X|}}{\sqrt 2\tanh^2(l/4)}~. $$
\end{lemma}

\begin{proof}
    The proof is based on a classical bound on the Schwarzian derivative at the center of a holomorphic map which is injective on a disk, and on the fact that if the shortest compressible curve in $(\partial X,h)$ has length at least $l$, then every point in $(\partial \tilde{X}, h)$ is the center of an embedded open disk of radius $l/2$. This second point follows from the fact that if $x\in \partial \tilde{X}$ realizes the minimum of the injectivity radius (the radius $r$ of the largest embedded disk centered at $x$) then there is a embedded open disk of radius $r$ centered at $x$ with a self-tangency, and it follows that $x$ is on a closed geodesic of length $2r$ in $(\partial \tilde{X}, h)$. This closed geodesic projects to a closed compressible geodesic of length $2r$ on $\partial X$.

    We first note that if $\DD_r$ is the disk of radius $r$ in $\C$ and $f:\DD_r\to \C$ is a univalent holomorphic map, then $\cS(f)$ can be written as $\sigma dz^2$, with $|\sigma(0)|\leq 6/r^2$. Indeed, the function $\bar f:z\mapsto f(rz)$ is then a univalent holomorphic map from $\DD$ to $\C$, so that by the Nehari-Kraus estimate \cite{nehari-bams}, its Schwarzian differential can be written as $\cS(\bar f)=\bar\sigma(z)dz^2$, with $|\bar \sigma(0)|\leq 6$. But it follows from the definition of the Schwarzian derivative that
  $$ \cS(\bar f)=r^2\cS(f)~. $$

  Since the hyperbolic metric at $0$ is $4(dx^2+dy^2)$, the norm of the real part of $\cS(f)$ with respect to the hyperbolic metric $h$ on $\DD$ is bounded (pointwise) by
  $$ \| \mathrm{Re}(\cS(f)(0)) \|_h\leq \frac{3\sqrt{2}}{2r^2}~. $$

  Now let $x\in \partial_\infty \tilde{X}$ be a point where the injectivity radius is at least $l/2$. Consider the Riemann uniformization map $f$ from $\DD$ to the connected component of $\partial_\infty\tilde X$, chosen so that $f(0)=x$. By construction, $f$ is a local isometry between the hyperbolic metric on $\DD$ and the Poincar\'e metric on $\partial_\infty \tilde{X}$. Moreover, the disk of center $x$ and radius $l/2$ (for the Poincar\'e metric) is embedded in $\partial_\infty\tilde X$, so that the restriction of $f$ to a disk of center $0$ and hyperbolic radius $l/2$ is injective. But a disk of hyperbolic radius $l/2$ and center $0$ (for the hyperbolic metric on $\DD$, that is, the Poincar\'e disk model of the hyperbolic plane) is a disk of Euclidean radius $r=\tanh(l/4)$. So the norm of the real part of the Schwarzian derivative of $f$ at $x$ is bounded, in the hyperbolic metric, by
  $$ \| \mathrm{Re}(\cS(f))\|_h \leq \frac{3\sqrt 2}{2\tanh^2(l/4)}~. $$
  Integrating over $\partial_\infty X$, we obtain that the $L^2$ norm of $\mathrm{Re}(q)$ is bounded by
  $$ \left(\int_S \| \mathrm{Re}(q)\|_h^2 da_h\right)^{1/2}\leq \frac{3\sqrt{\pi|\chi(\partial_\infty X)|}}{\tanh^2(l/4)}~. $$
  It then follows from Lemma \ref{lm:gradient} that
  $$ \| dV_R\|_{WP}\leq \frac{3\sqrt{\pi|\chi(\partial_\infty X|}}{\sqrt 2\tanh^2(l/4)}~. $$
\end{proof}

We notice for future reference that the WP estimate here could be improved, since the pointwise estimate on $\mathrm{Re}(q)$ is better at each point where the injectivity radius is larger than $l/2$.

It follows from this lemma that on any 1-parameter family of boundary complex structures $(c_t)_{t\in [0,1)}$ in $\cT_{\partial X}/\Gamma$ of finite Weil-Petersson length, ending on a stratum of the Weil-Petersson completion of $\cT_{\partial X}/\Gamma$ corresponding to pinching a closed curve which is not compressible in $X$, $V_R$ remains bounded (see \cite{vargas-pallete-continuity}). This follows from the lemma since, in this 1-parameter family, the lengths of simple, non-trivial closed curves compressible in $X$ remains bounded from below by a positive constant.

\subsection{Convergence of convex cores when pinching compressible curves}
\label{ssc:convexcore}

In this section we consider a sequence of conformal structures $(c_n)_{n\in \N}$ on $\partial X$, $c_n\in \cT_{\partial X}/\Gamma$, and denote by $h_n$ the hyperbolic metric in the conformal class $c_n$. We assume that $(c_n)_{n\in \N}$ converges, in the Weil-Petersson metric completion of $\cT_{\partial X}/\Gamma$, to a boundary point $c_\infty$, where a set of disjoint simple closed curves $\gamma_i, 1\leq i\leq k$ is pinched, with each $\gamma_i$ bounding a disk $D_i$ in $X$. Let $X_j, 1\leq j\leq l$, be the connected components of $X\setminus (D_1\cup \cdots \cup D_k)$. 

We denote by $g_n$ the convex co-compact hyperbolic metric on $X$ associated to $c_n$ by the Ahlfors-Bers theorem, and by $X(n)=(X, g_n)$. The limit conformal structure $c_\infty$ determines a conformal structure on $\partial X_j, 1\leq j\leq l$, marked by $2k$ points $\xi_1, \cdots, \xi_{2k}$ corresponding to the pinching of the $\gamma_i, 1\leq i\leq k$, and we denote by $\xi(j)$ the set composed of the $\xi_i$ which are in the boundary of $X_j$, $1\leq j\leq l$. (Note that if $X$ is connected, then each of the $X_j$ is equipped with at least one of the $\xi_i$ on its boundary, so $\xi(j)$ has at least one element.)

Now consider $\epsilon>0$ small enough so that for $n\geq n_0$, for a $n_0$ large enough, the $\gamma_i$ have length less than $\epsilon$ for $1\leq i\leq k$, while all closed curves not homotopic to a finite cover of one of the $\gamma_i$ have length larger than $\epsilon$. For $n\geq n_0$, $\partial C(X(n))$ is the union of:
\begin{itemize}
\item $k$ tubes $T_i(n), 1\leq i\leq k$, composed of points where the injectivity radius is less than $\epsilon/2$ -- each tube having  the geodesic representative of one of the $\gamma_i$ as its core meridian,
\item $l$ connected regions $C_j(n), 1\leq j\leq l$, one for each of the $X_j$, composed of points where the injectivity radius is at least $\epsilon/2$. 
\end{itemize}
The diameter of each of the $C_j(n)$, $1\leq j\leq l$, is uniformly bounded from above (independently of $n$), because $\partial C(X(n))$ has bounded area by the Gauss-Bonnet theorem, and the $C_j(n)$ are connected and composed of points were the injectivity radius is at least $\epsilon/2$.

We choose in each of the $C_j(n), 1\leq j\leq l$, a point $x_j(n)$. The $x_j(n)$ will be used as base-points below, when considering the Gromov-Hausdorff convergence of the $X_j(n)$. The precise choice of the $x_j(n)$ is not important, since they are constrained to be contained in a region of bounded diameter -- this will be sufficient to ensure convergence in the Gromov-Hausdorff topology pointed at $x_j(n)$, after extracting a sub-sequence.


\begin{lemma} \label{lm:convexcore}
  Under the hypothesis above, for each $j\in \{ 1,2,\cdots, l\}$ (and after extracting a subsequence):
  \begin{enumerate}
  \item The pointed manifold $(X(n),x_j(n))$ converges in the Gromov-Hausdorff topology on compact subsets to a complete hyperbolic manifold $(\bar X_j, \bar x_j)$, with $\bar X_j$ diffeomorphic to $X_j$. ($\bar X_j$ can be either convex co-compact, a solid torus, or a ball.)
  \item The sequence of pointed convex cores $(C(X(n)), x_j(n))$ converges to $CH(\xi(j))\subset \bar X_j$, the convex hull of $\xi(j)$ in $\bar X_j$.
  \item $V_C(X(n))\to \sum_{j=1}^l V(CH(\xi(j)))$, which is finite.
  \end{enumerate}
\end{lemma}

Here by the {\em convex hull} of $\xi(j)$ we mean the smallest geodesically convex subset of $\bar X_j$ containing $\xi(j)$ in its asymptotic boundary. By definition this subset contains the convex core of $\bar X_j$. 

\begin{proof}
  We will use the same argument as in Section \ref{ssc:maximal}, and notice that there exists $n_0\in \N$ such that for $n\geq n_0$ the lengths of all the $\gamma_i$ is at most $\epsilon'$, for a fixed value of $\epsilon'>0$. If $\epsilon'$ is small enough, we can then consider a totally geodesic plane $P_i$ as in Section \ref{ssc:maximal}, that is, a plane orthogonal to one of the lines in the thin tube with core meridian $\gamma_i$. The planes $P_i$ and $P_{i'}$ are then disjoint for $i\neq i'$, as seen in Section \ref{ssc:maximal}. We now only consider $n\geq n_0$.

  We then let $\Omega_j(n)$ be the connected component of $X(n)\setminus (P_1\cup\cdots\cup P_{k})$ containing $x_j(n)$, and let $K_j(n)=\Omega_j(n)\cap C(X,g_n)$. As in Section \ref{ssc:maximal}, we glue a half-space to $\Omega_j(n)$ at each of the $P_i$ adjacent to $\Omega_j(n)$, and obtain in this manner a complete hyperbolic manifold $X_j(n)$, which is either convex co-compact, a solid torus, or a ball. Abusing notations a bit, we consider the $P_i$ adjacent to $\Omega_j(n)$ as disjoint, totally geodesic planes in $X_j(n)$.

  To this surgery in the hyperbolic metric corresponds a simple surgery on the conformal structure at infinity of $X(n)$: a curve (which is short in say the Poincar\'e metric on $\partial_\infty X(n)$) is cut and one side replaced by a small disk. As $n\to \infty$, the conformal structure at infinity obtained in this manner converges to the conformal structure $c_\infty$ on $\partial X_j$. So $X_j(n)$ converges in the Gromov-Hausdorff topology to $\bar X_j$. This proves the first point.
  
  The closure of the convex subset $K_j(n)$ is the convex hull in $X_j(n)$ of the $P_i\cap \partial \Omega_j(n)$, which are topological disks with boundaries corresponding to the $\gamma_i$. This is clear because $K_j(n)$ is geodesically convex by definition, and its boundary is a pleated surface outside of the $P_i\cap \partial K_j(n)$, so its closure is the minimal closed geodesically convex subset of $X_j(n)$ containing the $P_i\cap \partial K_j(n)$. The boundary of $P_i\cap \partial \Omega_j(n)$ corresponds to $\gamma_i$ and its length goes to $0$ as $n\to \infty$. Moreover, for each $i$, $d(x_j(n), P_i\cap \partial \Omega_j(n))\to \infty$ as $n\to \infty$, because the length of the tube $T_i(n)$ around $\gamma_i$ goes to infinity. As each of the  $P_i\cap \partial \Omega_j(n)$ converges to one of the points of $\xi(j)$ as $n\to \infty$, we see that $K_j(n)$ converges to the convex hull of $\xi(j)$ in $\bar X_j$, which proves point (2). 

  The convergence of $V_C(X(n))$ to the sum of the volumes of the $CH(\xi(j))$ follows from the Gromov-Hausdorff convergence of the different components of $X(n)$, pointed at the $x_j(n)$.

  To see that the volume of $CH(\xi(j))$ is bounded, consider $\xi\in \xi(j)$, let $P$ be a totally geodesic plane in $\bar X_j$ separating $\xi$ from $C(\bar X_j)$ and from the other elements of $\xi(j)$, and let $H$ be the half-space bounded by $P$ containing $\xi$ in its boundary. Then $H\cap CH(\xi)$ is the convex hull of $\xi$ and of a compact domain in $P$ (the intersection of $CH(\xi(j))$ with $P$). So  $H\cap CH(\xi)$ has finite volume. Since this holds for all the $\xi\in \xi(j)$, we see that $CH(\xi(j))$ can be written as the union of a finite family of subsets of finite volume -- one for each element of $\xi(j)$ -- and the remaining part which is compact. So $CH(\xi(j))$ has finite volume.
\end{proof}

\subsection{Limit of the renormalized volume when pinching a compressible curve}

We now consider the case where $(c_t)_{t\in [0,1)}$ pinches a curve $\gamma$ which is contractible in $X$. The fact that $V_R(c_t)\to -\infty$ as $t\to 1$ then follow from Lemma \ref{lm:compare} and Lemma \ref{lm:convexcore}, together with the following result of Bridgeman and Canary, see \cite[Theorem 2']{bridgeman-canary:bounding}, and also \cite[Theorem 4.2]{bridgeman-canary:renormalized}.

\begin{theorem}[Bridgeman-Canary]
  There exists constants $P$ and $Q$ (one can take $P=74$ and $Q=36$) such that if $X$ is a convex co-compact hyperbolic manifold such $\partial_\infty\tilde X$ contains a closed compressible geodesic of length $r<1$ in the Poincar\'e metric, then the length of the measured bending lamination on the boundary of the convex core is bounded from below by:
  \begin{equation}
    \label{eq:Lml}
  L_m(l) \geq \frac Pr -Q~.   
  \end{equation}
\end{theorem}

Note that the coefficient $74$ above is twice that found in \cite{bridgeman-canary:renormalized}, since we consider here the length of closed contractible curves in $\partial X$, rather than the injectivity radius of $\partial \tilde X$.

Without getting into the precise value of the constants, we can indicate a heuristic explanation for \eqref{eq:Lml}. When $\partial_\infty X$ contains a closed geodesic $\gamma$ of length $r$ for the Poincar\'e metric $h$, then it contains a collar of width approximately $L=|\log(r)|$ around $\gamma$. If $\gamma$ is contractible, the induced metric $m$ on $\partial C(X)$ contains a tube of length $w$ approximately $\exp(L)=1/r$ around $\gamma$. But since $\gamma$ is contractible, the intersection with $\gamma$ of the measured bending lamination $l$ (that is, the transverse measure of $l$ evaluated on $\gamma$) is at least $2\pi$, and in fact very close to $2\pi$ as $r\to 0$. Finally, all leaves of $l$ intersecting $\gamma$ must cross the whole length of the tubular collar around $\gamma$, so must have length at least $2w$, so of the order of $1/r$. More precise arguments of this type can lead to \eqref{eq:Lml}.



\subsection{Geodesically convex subsets of a hyperbolic manifold}
\label{ssc:A-convex}

We first indicate why the intersection of two closed, non-empty, geodesically convex subsets $K, K'$ in a complete hyperbolic manifold $X\neq \HH^3$ is non-empty. Since $X$ has non-trivial topology and is complete, it contains a closed, oriented geodesic $\gamma$. Let $x\in K$, and, for $n\in \N$, let $\gamma_n$ be the geodesic segment starting and ending at $x$, and homotopic to a path going from $x$ to a point of $\gamma$, doing $n$ turns around $\gamma$, and going back to $x$.

Let $\bar \gamma_n$ and $\bar \gamma$ be lifts of $\gamma_n$ and of $\gamma$ to the universal cover of $X$, chosen so that the distance from the endpoints of $\bar \gamma_n$ to $\bar \gamma$ is equal to the distance of $x$ to $\gamma$ in $X$. For $n$ large, $\bar\gamma_n$ is a long geodesic segment with endpoints at bounded distance from $\bar \gamma$, so there is a sub-segment of $\bar\gamma_n$ of length at least the length of $\gamma$ which is arbitrarily close to $\bar \gamma$, say at distance less than $\epsilon_n$, for some $\epsilon_n>0$ with $\lim_{n\to \infty}\epsilon_n=0$.

It follows that  each point of $\gamma$ is at distance at most $\epsilon_n$ from a point of $\gamma_n$. Since $K$ is geodesically convex, $\gamma_n\subset K$, and since $K$ is closed and $\epsilon_n\to 0$, $\gamma\subset K$. The same argument shows that $\gamma\subset K'$, and it follows that $K\cap K'\neq\emptyset$.

Let now $K\subset X$ be a geodesically convex, let $x\in \partial K$, and let $n$ be the outward oriented unit normal to a support planes\footnote{A support plane of $K$ at $x$ is a totally geodesic plane containing $x$, which locally bounds a closed half-space containing $K$.} of $K$ at $x$. Let $\alpha$ be the half-geodesic starting from $x$ in the direction of $n$. Then $\alpha\cap K=\{ x\}$, since otherwise the whole intersection of $\alpha$ between $x$ and its first intersection with $K$ would be contained in $K$.

If $\beta$ is another such geodesic ray, starting from a point $x'\in \partial K$ in the direction of a unit normal vector to a support plane of $K$ at $x'$, then $\alpha$ and $\beta$ must be disjoint. Suppose indeed that they intersect at a point $y$, and let $\gamma$ be the geodesic segment from $x$ to $x'$ homotomic to the union of the segment of $\alpha$ from $x$ to $y$ union the segment of $\beta$ from $y$ to $x'$. The angle between $\gamma$ and $\alpha$ (resp. $\gamma$ and $\beta$) must be bigger than $\pi/2$, because both $\alpha$ and $\beta$ are directed by outwards unit normals of support planes of $K$, while $\gamma$ is towards the interior of $K$. But having two angles larger than $\pi/2$ contradicts the Gauss-Bonnet relation for hyperbolic triangles (the sum of the interior angles is equal to $\pi$ minus the area).

\begin{figure}[h]
  \centering
 \includegraphics[width = 4.0cm]{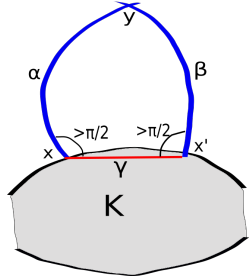}
 \caption{A pair of geodesic rays $\alpha,\beta$ normal to a convex subset $K$ cannot intersect, because the sum of the angles of triangle $xx'y$ would be
 greater than $\pi$.}
 \end{figure}  

 It follows from this remark that, if we denote again by $N\partial K$ the unit normal bundle of the boundary of $K$, then the map
 $$  \begin{array}{cccc}
       \exp: & N\partial K\times \R_{>0} & \to & X \\ 
            & (n,t) &\mapsto & \exp(nt)
 \end{array} $$
 is a diffeomorphism from $N\partial K\times \R_{>0}$ to $X\setminus K$.

 It also follows that $K$ is homotopic to $X$. This applies in particular to $C(X)$, the convex core of $X$.


\subsection{The Thurston metric at infinity}
\label{ssc:thurston-metric}

Finally, this section describes briefly some properties of a natural metric at infinity of convex co-compact hyperbolic manifolds, which appears prominently in Sections \ref{ssc:r-thurston} and \ref{ssc:asymptotic}.

The boundary at infinity $\partial_\infty X$ of a convex co-compact hyperbolic manifold is equipped naturally with a Riemannian metric in the standard conformal class, called the Thurston metric (or projective metric, or sometimes the grafting metric), closely related to the boundary of $C(X)$. We denote it by $h_{Th}$, and give three descriptions of it.

The first description is simpler when the measured bending lamination $l$ on $\partial C(X)$ is along disjoint closed geodesics $c_1, \cdots, c_n$, with each $c_i$ equipped with a positive weight $w_i$. In this case, $h_{Th}$ is obtained by cutting $(M,m)$ along the geodesics $c_i$ and replacing (or ``grafting'') $c_i$ by a flat strip of width $w_i$.

A second description is as the natural metric induced on $N\partial C(X)$, the unit normal bundle of the boundary of $C(X)$. As seen in Section \ref{ssc:boundary}, the normal exponential map is a homeomorphism between $N\partial C(X)$ and $\partial_\infty X$, so that the metric can then be pushed to $\partial_\infty X$.

A third description is as the metric at infinity defined by the equidistant foliation of $X\setminus C(X)$. For $r>0$, let $M_r$ be the set of points at distance $r$ from $C(X)$. The surfaces $M_r$ can be identified through the flow of the normal directions. Let $h_r$ be the induced metric on $M_r$. As $r\to \infty$, $h_r$ ``expands'' exponentially, but the ``normalized'' metric $4e^{-2r}h_r$ converges to $h_{Th}$.

A consequence of this last description is that the renormalized volume associated to the metric $h_{Th}$ at infinity is (up to an additive, topological constant)
$$ V_C(X)-\frac 14L_m(l)~, $$
where $V_C(X)$ is the volume of the convex core and $L_m(l)$ the length of the measured bending lamination on the boundary of the convex core.

\subsection{The renormalized volume associated to the Thurston metric}
\label{ssc:r-thurston}

In this section we show that the renormalized volume $V'_R$ of $X$ associated to the Thurston metric $h_{Th}$ at infinity is within a bounded additive constant from the renormalized volume $V_R$ associated to the Poincar\'e metric $h$ at infinity.

\begin{theorem} \label{tm:compare2}
  There exists a constant $C(\partial X)$, depending only on the topology of $\partial X$, such that
  $$ V'_R\leq V_R+C(\partial X)~. $$
\end{theorem}

Note that the opposite inequality $V_R\leq V'_R$ holds, up to an additive constant depending on normalization, as already stated in Lemma \ref{lm:compare}.

The heuristic idea of the proof of Theorem \ref{tm:compare2} is quite simple.
\begin{itemize}
\item The difference $V'_R-V_R$ can be expressed as an integral of a function of the conformal factor between the hyperbolic metric and the Thurston metric, see Definition \ref{df:polyakov}. Or, more specifically, in terms of the function $u$ such that $h_{Th}=e^{2u}h$. 
\item The contribution to this integral of the ``thick'' part of $\partial X$ -- the set of points where the injectivity radius for $h$ is bounded from below -- is uniformly bounded, because the conformal factor and its gradient are bounded in this region, see Lemma \ref{lm:C1}. So we can focus on the long thin tubes around closed geodesics which are short for $h$.
\item On those tubes, the Thurston metric can be approximated by a flat metric on a long cylinder of perimeter $2\pi$. We approximate the Thurston metric $h_{Th}$ by such a flat metric $h_\gamma$, and write $u=u_0+u_\Delta$, where $h_\gamma=e^{2u_0}h$, and $h_{Th}=e^{2u_\Delta}h_\gamma$.
\item The integral term corresponding to $u_0$ can then be explicitly computed (see Lemma \ref{lm:h-h_gamma}) and it is bounded.
\item Moreover, one can find sufficient bounds on $u_\Delta$ to show that the correction coming from $u_\Delta$ is also bounded (see Lemma \ref{lm:C2}).
\end{itemize}

We now proceed with the proof.

\begin{definition} \label{df:polyakov}
  Let $S\subset \partial X$, let $g$ be a Riemannian metric on $S$, and let $u:S\to \R$ be a function. Let
  $$ W_S(e^{2u}g, g) = -\frac 14\int_S (\| du\|_{g}^2 - 2K_{g}u) da_{g}~. $$
\end{definition}

  It follows from the ``Polyakov formula'' for the dependence of the renormalized volume on the metric at infinity that 
$$ V'_R-V_R = W_{\partial X}(h_{Th},h)~. $$
Moreover $h_{Th}$ is conformal to the Poincar\'e metric $h$, so we can write
$$ h_{Th}=e^{2u} h~, $$
for a function $u:\partial X\to \R$.

We will use the following well-known $C^1$ bound on $u$ in the ``thick'' part of $x\in (\partial X, h)$.

\begin{lemma} \label{lm:C1}
  There exists a constant $C_1>0$ such that for all $x\in (\partial X, h)$ where the injectivity radius of $h$ is at least $\epsilon_0/2$, $u\leq C_1$ and $\| du\|_h\leq C_1$.
\end{lemma}

\begin{proof}
  The bound on $u$ can be found in the proof of Theorem 2.17 in \cite{bridgeman-brock-bromberg}.

  The bound on $du$ then follows from the bound on $u$ and the bound on $\Delta u$ together with the lower bound on the injectivity radius. Although this estimate is well-known to analyst, we include an informal argument for completeness. Let $x$ be such a point, and let $r>0$ be such that the disk $D(x,r)$ of center $x$ and radius $r$ is embedded. We can write the restriction of $u$ to $D(x,r)$ as
  $$ u=v+w~, $$
  with
  $$ v_{|\partial D(x,r)}=u_{|\partial D(x,r)}~, ~~\Delta_h v=0~, $$
  $$ w_{|\partial D(x,r)}=0~, ~~ \Delta_hw=\Delta_hu~. $$
  Then $w$ can be written as an integral over $D(x,r)$ of Green functions for $\Delta_h$ on $D(x,r)$, multiplied by $\Delta_hu$, which is uniformly bounded. The uniform bound on $dw$ at $x$ follows. In addition, $dv$ is uniformly bounded at $x$ because $u$ (and therefore also $v$) is uniformly bounded on $\partial D(x,r)$. The uniform bound on $du$ at $x$ follows.
\end{proof}

We denote by $T_\gamma$ the ``Margulis tube'' associated to $\gamma$ in the hyperbolic metric $h$. That is, we fix a constant $\epsilon_0>0$, and let $T_\gamma$ be the set of points at distance at most $r$ (from $h$) from the geodesic representative of $\gamma$, with $r$ chosen so that the boundary of $T_\gamma$ is the disjoint union of two closed curves of length $\epsilon_0$ -- this is possible if $\epsilon_0$ is small enough. This tube $T_\gamma$ can also be defined as the connected component of the geodesic representative of $\gamma$ in the set of points in $(\partial X,h)$ where the injectivity radius is at most $\epsilon_0/2$.

This tube is also equipped with a standard Euclidean metric $h_\gamma$ conformal to the restriction of $h$ to $T_\gamma$. We choose this metric $h_\gamma$ to be isometric to $S^1\times [-L_\gamma,L_\gamma]$, where $L_\gamma$ will be determined below.

\begin{lemma} \label{lm:h-h_gamma}
  There exists a constant $C_0$ such that $|W_{T_\gamma}(h_\gamma,h)|\leq C_0$.
\end{lemma}

\begin{proof}
  By definition, $h_\gamma$ is conformal to $h$, so there exists a function $u_0:T_\gamma\to \R$ such that
  $$ h_\gamma=e^{2u_0}h $$
  on $T_\gamma$, with $u_0$ constant on $\partial T_\gamma$.

  This function $u_0$ is clearly invariant by rotation, and only depends on the distance $z$ to the core curve of $T_\gamma$ (the geodesic representative of $\gamma$ for $h$). Since the curve composed of points at (oriented) distance $z$ from the core curve has length $l_\gamma \cosh(z)$ for $h$, while it has length $2\pi$ for $h_\gamma$. Since each boundary component of $T_\gamma$ has length $2\pi$ for $h$,
  $$ e^{u_0(z)} = \frac{2\pi}{l_\gamma \cosh(z)}~, $$
  and therefore
  $$ |u_0'(z)| = |- \tanh(z)|\leq 1~. $$

  As a consequence, 
  \begin{eqnarray*}
    |W_{T_\gamma}(h_\gamma,h)| & = & \frac 14\left|\int_{T_\gamma} \| du_0\|^2_h + 2u_0 \right| da_h \\
                               & \leq & \frac 14{\rm Area}(T_\gamma,h) + \frac 12 \int_{T_\gamma} u_0 da_h \\
    & \leq & \frac 14{\rm Area}(T_\gamma,h) + \int_0^{L_\gamma} l_\gamma\cosh(z) \log\left(\frac{2\pi}{l_\gamma\cosh(z)}\right) dz~,
  \end{eqnarray*}
  where $L_\gamma$ is the half-length of $T_\gamma$ for $h$, that is, such that
  $$ l_\gamma \cosh(L_\gamma)=\epsilon_0~. $$
  A direct computation shows that
  $$  \int_0^{L_\gamma} l_\gamma\cosh(z) \log(l_\gamma\cosh(z))dz = $$
  $$ = l_\gamma \sinh(L_\gamma)\log(l_\gamma \cosh(L_\gamma)) +2l_\gamma (\pi/4-\rm{arctan}(e^{-L_\gamma}))-l_\gamma \sinh(L_\gamma)~, $$
  and since $l_\gamma\cosh(L_\gamma)=\epsilon_0$,
  $$ \left|  \int_0^{L_\gamma} l_\gamma\cosh(z) \log(l_\gamma\cosh(z))dz \right| \leq 
  \epsilon_0\log(\epsilon_0) + \epsilon_0 + l_\gamma \frac \pi 2~. $$

  Finally,
  $$ \int_0^{L_\gamma} l_\gamma \cosh(z)dz = l_\gamma\sinh(L_\gamma)\leq \epsilon_0~. $$
  Adding the terms in the upper bound on $|W_{T_\gamma}(h_\gamma,h)|$ yields the result.
\end{proof}

We can also compare the flat metric $h_\gamma$ on $T_\gamma$ to the Thurston metric $h_{Th}$. Since $h_{Th}$ is conformal to $h$, it is also conformal to $h_\gamma$, so we can write
$$ h_{Th} = e^{2u_\Delta} h_\gamma~, $$
for a function $u_\Delta:T_\gamma\to \R$. By definition, $u=u_0+u_\Delta$ on $T_\gamma$. The next lemma states a bound on $u_\Delta$, with fixed constants, over $T_\gamma$.

\begin{lemma} \label{lm:C2}
  There exists a constant $C_2>0$ such that, on $T_\gamma$, 
  $$ |u_\Delta|\leq C_2 $$
  and
  $$ \int_{T_\gamma}\| du_\Delta\|^2_{h_{Th}}da_{h_{Th}}\leq C_2~. $$
\end{lemma}

\begin{proof}
  Since $h_\gamma = e^{-2u_\Delta} h_{Th}$ and $h_\gamma$ is flat, the curvature of $h_{Th}$ satisfies
  $$ \Delta_{Th}u_\Delta =K_{Th}~. $$
  Since $K_{Th}\leq 0$, $\Delta_{Th}u_{\Delta}\leq 0$, so $u_\Delta\geq C_1$ by the maximum principle.

  In addition, $\Delta_{Th} u_\Delta\geq -1$. It then follows from standard arguments (using the fact that $K_{Th}\in [-1,0]$) that there exists $c>0$ such that if $u_\Delta(x_0)\geq c$, then $u_\Delta\geq 1$ on the disk of center $x_0$ and radius $2\pi$ in $h_{Th}$. Since $(T_\gamma, h_{Th})$ is approximated, outside a neighborhood of its boundary, by a tube of perimeter $2\pi$, there exists a closed curve $\gamma'$ homotopic to $\gamma$ going through $x_0$, of length less than $3\pi$ for $h_{Th}$. Then $u_\Delta\geq 1$ on $\gamma'$. It would then follow that the length of $\gamma'$ for $h_\gamma$ is at most $3\pi e^{-2}<2\pi$, a contradiction since $(T_\gamma,h_\gamma)$ is isometric to a cylinder of perimeter $2\pi$ and $\gamma'$ is homotopic to $\gamma$. So $u_\Delta\leq c$ on $T_\gamma$. We can already conclude that, for a certain $C_2>0$, $ |u_\Delta|\leq C_2$ on $T_\gamma$.

  Notice that
  \begin{eqnarray*}
    \int_{T_\gamma}\| du_\Delta\|^2_{h_{Th}} da_{h_{Th}} & = & \int_{T_\gamma} u_\Delta\Delta_{Th} u_{\Delta} da_{h_{Th}} + \int_{\partial T_\gamma} u_\Delta du_\Delta(n) ds \\
    & = &\int_{T_\gamma} u_\Delta K_{Th} da_{h_{Th}} + \int_{\partial T_\gamma} u_\Delta du_\Delta(n) ds~, 
  \end{eqnarray*}
  so
\begin{eqnarray*}
  \left| \int_{T_\gamma}\| du_\Delta\|^2_{h_{Th}} da_{h_{Th}} \right| & \leq & C_2\left|\int_{T_\gamma} K_{Th} da_{h_{Th}}  \right| + \int_{\partial T_\gamma}\left|u_\Delta du_\Delta(n)\right|ds \\
                                                           & \leq & C_2 {\rm Area}(T_\gamma,m) + 2C_1^2\epsilon_0~.
\end{eqnarray*}
The uniform bound on the integral follows.
\end{proof}

\begin{proof}[Proof of Theorem \ref{tm:compare2}]
  Let $\gamma_1,\cdots, \gamma_n$ be the closed geodesics of length less than $\epsilon_0$ in $(\partial X, h)$. If $\epsilon_0$ is small enough, those short closed geodesics are disjoint, and they are the core curves of disjoint long thin tubes in $(\partial X,h)$, denoted here by $T_{\gamma_1}, \cdots, T_{\gamma_n}$. Since the $\gamma_i$ are disjoint, $n\leq (3/2)|\Chi(\partial X)|$. 

  For each $i\in \{ 1,\cdots, n\}$, we have seen that
  \begin{eqnarray*}
    \left|W_{T_{\gamma_i}}(h_{Th},h)\right| & = & \left|\frac 14\int_{T_{\gamma_i}}\| d(u_0+u_\Delta)\|^2_h+2(u_0+u_\Delta) da_h \right|~.
  \end{eqnarray*}
  Since $u=u_0+u_\Delta>0$, 
  \begin{eqnarray*}
    \left|W_{T_{\gamma_i}}(h_{Th},h)\right|
    & \leq & \left|\frac 12\int_{T_{\gamma_i}}\|_h du_0\|_h^2 + \|u_\Delta\|^2_h+ u_0+u_\Delta da_h \right| \\
    & \leq & \frac 12\left|\int_{T_{\gamma_i}}\| du_0\|^2_h+2u_0 da_h\right| +
             \frac 12\left|\int_{T_{\gamma_i}}u_0 da_h\right| +
             \frac 12\left|\int_{T_{\gamma_i}}\| du_\Delta\|^2_h + u_\Delta da_h\right| \\
    & \leq & 2|W_{T_{\gamma_i}}(h_{\gamma_i},h)| +
             \frac 12 \left|\int_{T_{\gamma_i}}u_0 da_h\right| +
             \frac 12\int_{T_{\gamma_i}}\| du_\Delta\|^2_{h_{Th}}da_{h_{Th}} +
             \frac 12  \int_{T_{\gamma_i}}u_\Delta da_h~. 
  \end{eqnarray*}
  However we have already seen in the proof of Lemma \ref{lm:h-h_gamma} that the integral of $u_0$ on $T_{\gamma_i}$ is bounded by a fixed constant, say $C'_0$.
  Using Lemma \ref{lm:h-h_gamma} and Lemma \ref{lm:C2}, it follows that
  \begin{eqnarray*}
    |W_{T_{\gamma_i}}(h_{Th},h)| & \leq & 2C_0 + \frac{C'_0}2 + \frac{C_2}2 + \frac{C_2}2{\rm Area}(T_{\gamma_i},h)~.  
  \end{eqnarray*}
  In addition
  $$ |W_{\partial X\setminus (T_{\gamma_1}\cup\cdots\cup T_{\gamma_n})}(h_{Th},h)|\leq \frac{C_1^2+C_1}4 {\rm Area}(\partial X,h)~, $$
  and the result follows.
\end{proof}

\subsection{Asymptotic behavior of the renormalized volume when pinching a compressible curve}
\label{ssc:asymptotic}

Finally we give here a more precise asymptotic description of the behavior of $V'_R$ when a compressible curve is pinched. This analysis can be extended to the case where two or more compressible curves are pinched, with a dominant term which is a sum of terms corresponding to each pinched curve.

\begin{theorem} \label{tm:asymptotics}
  Let $(c_t)_{t\in [0,1)}$ be a smooth curve in $\cT_{\partial X}$, with $\lim_{t\to 1}c_t$ a point in the Weil-Petersson compactification of $\cT_{\partial X}$ corresponding to a hyperbolic metric with one simple compressible closed curve $\gamma$ pinched. Then, as $t\to 1$,
  $$ V'_R(c_t) \sim \frac{-\pi^3}{L_{c_t}(\gamma)}~, $$
  where $L_{c_t}(\gamma)$ is the length of $\gamma$ in the hyperbolic metric $c_t$ on $\partial X$.
\end{theorem}

Note that the proof actually shows a little more: in the case where several curves are pinched so as to have (asymptotically) constant length ratio, $V'_R$ is equivalent to a sum of terms corresponding to each of those short curves.

In the next lemma, we consider $T_\gamma$ as a subset of $\partial C(X)$, using the nearest-point projection from $\partial_\infty X$ to $\partial C(X)$. 

\begin{lemma} \label{lm:ext}
  There exists a constant $C_3>0$ such that if the Margulis tube $T_\gamma$ contains a maximal segment of length $2L$ for the induced metric $m$ in the support of the measured bending lamination $l$, then all segments of $l$ in $T_\gamma$ have length in $[2L-C_3, 2L+C_3]$, and the extremal length of $\gamma$ satisfies
  $$ |\ext(\gamma) - \frac\pi L|\leq C_3~. $$
\end{lemma}

\begin{proof}
  Let $c$ and $c'$ be two maximal segments in $T_\gamma$ in the support of $l$. Let $c_-, c_+$ be the endpoints of $c$ on $\partial T_\gamma$, and similarly let $c'_-, c'_+$ be the endpoints of $c'$, with $c'_-$ on the same boundary component of $T_\gamma$ as $c_-$. Since the boundary components of $T_\gamma$ have length $\epsilon_0$ for $h$, $c_-$ and $c'_-$ are at distance at most $\epsilon_0/2$ for $h$, and similarly for $c_+$ and $c'_+$.

  Lemma \ref{lm:C1} therefore shows that $c_-$ and $c'_-$ are also at distance at most $e^{C_1}\epsilon_0/2$ in the Thurston metric $h_{Th}$, and similarly for $c_+$ and $c'_+$. But then it follows that they are also at distance at most $e^{C_1}\epsilon_0/2$ in the induced metric $m$ on the $\partial C(X)$, which is smaller than $h_{Th}$. This constant therefore also bounds their hyperbolic distance in $X$. 

  Since $c$ and $c'$ are geodesics for the hyperbolic metric on $X$, it follows from the triangle inequality that the hyperbolic length of $c$ and $c'$ are close:
  $$ |L(c')-L(c)|\leq e^{C_1}\epsilon_0~. $$
  This proves the first point.

  For the second point, let $c_0$ be a geodesic segment of $C(X)$, of length $2L$, centered at a point close to the geodesic representative of $\gamma$ in $\partial C(X)$. The same argument as above for $c$ and $c'$ shows that there exists a constant $C>0$ such that the orthogonal projection of $T_\gamma$ on the geodesic containing $c_0$ is within the set of points at distance at most $C$ from $c_0$ -- we denote this extended segment by $c_{0+}$. Conversely, this orthogonal projection of $T_\gamma$ contains the set of points of $c_0$ at distance at least $C$ from the endpoints, denoted here by $c_{0-}$.

  Let $N^1c_{0+}$ be the unit normal space to $c_{0+}$, that is, the set of unit vectors orthogonal to $c_{0+}$. We consider the exponential map $\exp_\infty:N^1c_{0+}\to \partial_\infty X$ sending a vector $n\in N^1c_{0+}$ to the endpoint at infinity of the geodesic ray defined by $n$. If $\epsilon_0$ is small enough, then $\exp_\infty$ is a diffeomorphism on its image, which is an annulus in $\partial_\infty X$ containing $T_\gamma$.

  An explicit computation -- using for instance the Poincar\'e model of $\HH^3$ -- shows that $\mod(\exp_{\infty} (c_{0-}))=\frac{L-C}\pi$ while  $\mod(\exp_{\infty} (c_{0+}))=\frac{L+C}\pi$. Since the modulus is increasing under inclusion, it follows that
  $$ \frac{L-C}\pi \leq \mod(T_\gamma,h)\leq \frac{L+C}\pi~. $$
  
  Since the conformal structure in the ``thick'' part of $\partial X$ remains bounded, standard arguments then show that the modulus of $\gamma$ in $\partial X$ differs from its modulus in $T_\gamma$ by bounded quantity, that is, replacing $C$ if necessary, any annulus $A$ extending $T_\gamma$ in $\partial X$ satisfies
  $$ \mod(T_\gamma,h) \leq \mod(A,h)\leq \mod(T_\gamma,h)+C~. $$
  
  Therefore the extremal length of $\gamma$ in $\partial X$ satisfies
  $$ \frac\pi{L+2C} \leq \ext_{\partial X}(\gamma) \leq \frac\pi{L-C}~, $$
  and the result follows.
\end{proof}

\begin{proof}[Proof of Theorem \ref{tm:asymptotics}]
  We assume that the hyperbolic length of $\gamma$ goes to zero as $t\to 1$. It then follows from results of Maskit \cite{maskit} that the extremal length of $\gamma$ for $c_t$ is equivalent to
  $$ \ext_c(\gamma) \sim \frac{L_{c_t}(\gamma)}\pi~. $$
  It then follows from Lemma \ref{lm:ext} that if the Margulis tube around $\gamma$ contains a maximal segment of length $2L_t$ in the support of the measured bending lamination, then
  $$ L_t \sim \frac\pi{\ext_{c_t}(\gamma)} \sim \frac{\pi^2}{L_{c_t}(\gamma)}~. $$
  However, the total length of the measured bending lamination in the ``thick'' part of $\partial C(X)$ (the set of points where the injectivity radius is larger than $\epsilon_0/2$) is uniformly bounded.
  Moreover, the intersection with $\gamma$ of the  measured bending lamination $l_t$ on the boundary of the  convex core converges to
  $$ i(\gamma, l_t)\to 2\pi~, $$
  as can be seen by considering the intersection of $\partial C(X)$ with the plane $P$ considered in Section \ref{ssc:maximal}.
  As a consequence, since the length of every segment of $l_t$ in $T_\gamma$ has length approximatively $2L_t$, the length of $l_t$ for the induced metric $m_t$ on the convex core  behaves has
  $$ L_{m_t}(l_t) \sim 2\pi. 2L_t = 4\pi L_t~. $$
and, as a consequence,
  $$ V'_R \sim -\frac{L_{m_t}(l_t)}4 \sim - \pi L_t \sim -\frac{\pi^3}{L_{c_t}(\gamma)}~. $$ 
\end{proof}


 \bibliographystyle{unsrt}

\end{document}